\documentclass[12pt]{article}
\usepackage[english]{babel}
\usepackage{amsmath,amssymb,graphicx,enumerate,latexsym,theorem}
\usepackage[a4paper,top=2cm,bottom=2cm,left=1.5cm,right=1.5cm,marginparwidth=1.75cm]{geometry}
\catcode`\<=\active \def<{
\fontencoding{T1}\selectfont\symbol{60}\fontencoding{\encodingdefault}}
\catcode`\>=\active \def>{
\fontencoding{T1}\selectfont\symbol{62}\fontencoding{\encodingdefault}}
\newcommand{\assign}{:=}
\newcommand{\mathD}{\mathrm{D}}
\newcommand{\mathd}{\mathrm{d}}
\newcommand{\nobracket}{}
\newcommand{\tmdummy}{$\mbox{}$}
\newcommand{\tmmathbf}[1]{\ensuremath{\boldsymbol{#1}}}
\newcommand{\tmop}[1]{\ensuremath{\operatorname{#1}}}
\newcommand{\tmscript}[1]{\text{\scriptsize{$#1$}}}
\newcommand{\tmstrong}[1]{\textbf{#1}}
\newcommand{\tmtextit}[1]{\text{{\itshape{#1}}}}
\newenvironment{itemizedot}{\begin{itemize} }{\end{itemize}}
\newenvironment{itemizeminus}{\begin{itemize} }{\end{itemize}}
\newenvironment{proof}{\noindent\textbf{Proof\ }}{\hspace*{\fill}$\Box$\medskip}
\newcounter{nnacknowledgments}

{\theorembodyfont{\rmfamily}\newtheorem{acknowledgments*}[nnacknowledgments]{Acknowledgments}}
\newcounter{nnconvention}

{\theorembodyfont{\rmfamily}\newtheorem{convention*}[nnconvention]{Convention}}
\newtheorem{definition}{Definition}[subsection]
\newcounter{nndefinition}

\newtheorem{definition*}[nndefinition]{Definition}
{\theorembodyfont{\rmfamily}\newtheorem{example}{Example}[subsection]}
\newtheorem{lemma}{Lemma}[subsection]
\newcommand{\nonconverted}[1]{\mbox{}}
\newtheorem{proposition}{Proposition}[subsection]
{\theorembodyfont{\rmfamily}\newtheorem{remark}{Remark}[subsection]}
\newcounter{nnremark}

{\theorembodyfont{\rmfamily}\newtheorem{remark*}[nnremark]{Remark}}
\newtheorem{theorem}{Theorem}[subsection]
\numberwithin{equation}{section}

\begin{document}

\title{Perturbative BV-BFV formalism with homotopic renormalization: a case
study}

\author{Minghao Wang and Gongwang Yan}

\maketitle

\begin{abstract}
  We report a rigorous quantization of topological quantum mechanics on
  $\mathbb{R}_{\geqslant 0}$ and $\mathbf{I}= [0, 1]$ in perturbative BV-BFV
  formalism. Costello's homotopic renormalization is extended, and
  incorporated in our construction. Moreover, BV quantization of the same
  model studied in previous work {\cite{wangyantqm202203}} is derived from the
  BV-BFV quantization, leading to a comparison between different
  interpretations of ``state'' in these two frameworks.
\end{abstract}

{\tableofcontents}

\section{Introduction}

Among mathematical frameworks for quantum field theory (QFT), the one
developed by Costello {\cite{costellorenormalization}} is a candidate to
systematically study perturbative gauge theories. This framework uses
Batalin-Vilkovisky (BV) formalism {\cite{BATALIN198127}} to describe gauge
symmetry, and contains a procedure which we call ``homotopic renormalization''
to deal with UV problem. These two ingredients are combined together in the
language of homological algebra. A successful example quantized in this formalism is BCOV(Bershadsk-Cecotti-Ooguri-Vafa) theory, see {\cite{Costello:2012cy}}{\cite{Costello:2015xsa}}{\cite{Costello:2019jsy}}. The corresponding string field theory at all topology can only be quantized by this method so far, which is a degenerate BV theory coupled with large N gauge field.

However, this framework needs modification when the spacetime manifold has
boundary. In order to define the theory still within BV formalism, we need to
impose boundary condition and work on a ``restricted bulk field space'' only
(see e.g.,
{\cite{rabinovich2021factorization,wangyantqm202203,gwilliam2021factorization,gwilliam2019onelpqtzncs,zeng2021monopole}}
for this approach). In this setting, renormalization is less developed than
the closed spacetime case, and we can find discussions for specific models,
for examples, in {\cite{albert2016heatmfdwtbdr,rabinovich2021factorization}}.

Alternatively, we could generalize BV formalism itself to study QFT on
manifold with boundary. A potential generalization is proposed by Cattaneo,
Mnev and Reshetikhin {\cite{cattaneo2014classical,cattaneo2018perturbative}},
called ``BV-BFV formalism''\footnote{``BFV'' for Batalin-Fradkin-Vilkovisky
{\cite{batalin1983generalizedbf0904,fradkin1975quantization0904}}.}. Its
central notion, called ``modified quantum master equation (mQME)'',
characterizes that the anomaly to quantize the theory in BV formalism is
controlled by certain boundary data. This formalism also contains formulation
for (gauge) field theories on manifold with corners, hence may help with
topics such as functorial QFT and bulk-boundary correspondence.

While BV-BFV formalism has been shown fruitful even only at classical level
(see e.g.,
{\cite{cattaneo2016bvbfvgreh,canepa2019fullybvbfvgr3d,schiavina2015bvthesisbvbfvgr,rejzner2021asymptotic,martinoli2022bvplykvngthy}}),
it has not incorporated a systematic renormalization for quantization. For topological field theories, a successful example in quantum level can be found in {\cite{Cattaneo_2020}}. For field theories which are not topological, we need to add counter terms to make the contributions of Feynman graphs finite. Counter terms on manifolds with boundaries in homotopic renormalization has been discussed in {\cite{rabinovich2021factorization}}. To quantize field theories which are not topological in BV-BFV formalism, We hope to adapt Costello's homotopic renormalization to BV-BFV formalism. As the first step, we would like to clarify the relation between BV-BFV formalism and the approach within BV formalism mentioned above.

\subsection{Main results}

We use AKSZ type {\cite{1997aksz}} topological quantum mechanics (TQM) as the
toy model to study the above questions.

Based on previous work {\cite{wangyantqm202203}}, we obtain a rigorous BV-BFV
description of TQM on $\mathbb{R}_{\geqslant 0}$ and $\mathbf{I}= [0, 1]$,
with homotopic renormalization incorporated. The mQME's are stated in
Definition \ref{mqmetqmdfntn36rlylstlbl} and Definition
\ref{mqmeintvlcsdfntn42plslstlbl}. Their generic solutions are described in
Theorem \ref{thm310gnrcsltntqmrplus} and Theorem
\ref{ppstn42gnrcsltntqmintvllbbububu}, respectively. Then we derive and
sharpen the BV description in {\cite{wangyantqm202203}} from these BV-BFV
constructions (see Proposition \ref{thm51bvdscrptntqmrpstvlstplsend5151} and
Proposition \ref{qmesltntqmintvl56endendend}). We use $1$D BF theory to
demonstrate our result in Example \ref{eg37thycnstctn},
\ref{eg52savemestpstp}.

This brings the study on TQM to a new stage, and we would like to present the
result for TQM on $\mathbb{R}_{\geqslant 0}$ (Theorem
\ref{thm310gnrcsltntqmrplus} and Proposition
\ref{thm51bvdscrptntqmrpstvlstplsend5151}) here. Concretely, the TQM is of
AKSZ type with target being a finite dimensional graded symplectic vector
space $V$. A Lagrangian decomposition $V = L \oplus L'$ is chosen. We have:
\begin{itemizedot}
  \item Given a functional $I_0 = \pi^{\ast} (J^{\partial}) +
  \int_{\mathbb{R}_{\geqslant 0}} I^{\partial}$ with certain $I^{\partial} \in
  \mathcal{O} (V) [[\hbar]], J^{\partial} \in \mathcal{O} (L) [[\hbar]]$,
  \footnote{$\mathcal{O} (V), \mathcal{O} (L)$ are function rings on $V, L$,
  respectively. $\hbar$ is the formal ``quantum parameter''. $\pi^{\ast}
  (J^{\partial})$ is a boundary term and $\int_{\mathbb{R}_{\geqslant 0}}
  I^{\partial}$ denotes the $\Omega^{\bullet} (\mathbb{R}_{\geqslant
  0})$-linear extension of $I^{\partial}$ followed by integration over
  $\mathbb{R}_{\geqslant 0}$.} it induces a consistent BV-BFV interactive
  theory with polarization $V = L \oplus L'$ and a BFV operator
  $\tmmathbf{\Omega}_{L'}^{\tmop{left}} (H^{\partial}, -)$ if and only if
  \[ I^{\partial} \star_{\hbar} I^{\partial} = 0 \quad \text{and} \quad
     H^{\partial} = - e^{J^{\partial} / \hbar} \star_{\hbar} I^{\partial}
     \star_{\hbar} e^{- J^{\partial} / \hbar}, \]
  where $\star_{\hbar}$ denotes the Moyal product on $\mathcal{O} (V)
  [[\hbar]]$, and $\tmmathbf{\Omega}_{L'}^{\tmop{left}} (H^{\partial}, -)$
  denotes the Weyl quantization of $H^{\partial} \in \mathcal{O} (V)
  [[\hbar]]$ on $\mathcal{O} (L') [[\hbar]]$.
  
  \item The functional $I_0 = \pi^{\ast} (J^{\partial}) +
  \int_{\mathbb{R}_{\geqslant 0}} I^{\partial}$ above induces a consistent BV
  interactive theory with ``boundary condition $L$'' if and only
  if\footnote{By imposing ``boundary condition $L$'' we define a ``restricted
  bulk field space'' $\mathcal{E}_L \subset \mathcal{E}$ in
  (\ref{rstctdfldspctqmrpstvsctn3fml4}).}
  \[ I^{\partial} \star_{\hbar} I^{\partial} = 0, \quad \text{and} \quad
     \tmmathbf{\Omega}_L^{\tmop{right}} (e^{J^{\partial} / \hbar},
     I^{\partial}) = 0, \]
  where $\tmmathbf{\Omega}_L^{\tmop{right}} (-, I^{\partial})$ denotes the
  Weyl quantization of $I^{\partial}$ on $\mathcal{O} (L) [[\hbar]]$.
\end{itemizedot}
We would like to stress the following aspects of the story, which should
persist in more general settings.

\subsubsection*{Homotopic renormalization in BV-BFV formalism}

Perturbative BV-BFV formalism involves a BV structure, a splitting and a BFV
operator which all need to be properly regularized (or, renormalized) in order
to rigorously quantize a generic theory. For TQM, we only need to solve this
problem for the former two structures.

A renormalized BV structure has been constructed in
{\cite{rabinovich2021factorization}} using homotopic renormalization, endowing
the ``restricted bulk field space'' $\mathcal{E}_L$ with a differential BV
algebra $(\mathcal{O} (\mathcal{E}_L), \mathd, \partial_{K_t})$ for each
renormalization scale $t > 0$, see (\ref{rnmlzdbvel36fmllstlbl}). As for the
splitting, we propose a notion:
\begin{itemizedot}
  \item The {\tmstrong{renormalized splitting}} {\tmstrong{(at scale $t$)}} is
  the map $\theta_t : L' \rightarrow \mathcal{E}$ determined by
  \[ \theta_t (l') = 2 (\mathbb{I} (\tmmathbf{\alpha}) (l' \otimes -) \otimes
     1) \bar{P} (0, t) |_{C_1} \]
  for $l' \in L'$ (details see Definition \ref{dfntn34rnzdsplttglstlt}).
\end{itemizedot}
This is defined according to the renormalized BV structure associated to
$\mathcal{E}_L$, as it depends on the ``propagator from scale $0$ to scale
$t$'' $\bar{P} (0, t)$. $\theta_t$ flows with $t$ in a way compatible with the
homotopic renormalization group flow (\ref{cnjgtnfreebvalgtqmrstctfldspc},
\ref{bvbfvtotalfldspchmtpyrgflowfree}) of relevant BV structures:
\begin{itemizedot}
  \item For $\forall \varepsilon, \Lambda > 0$, the renormalized splittings
  $\theta_{\varepsilon}, \theta_{\Lambda}$ make this diagram commute:
  \begin{eqnarray*}
    \mathcal{O} (\mathcal{E}) [[\hbar]] \quad &
    \xrightarrow{\LARGE{\mathbb{I}_{\theta_{\varepsilon}}}} & \quad
    \mathcal{O} (L') \otimes \mathcal{O} (\mathcal{E}_L) [[\hbar]]\\
    \downarrow \enspace e^{\hbar \partial_{P (\varepsilon, \Lambda)}} &  &
    \hspace{3em} \downarrow \enspace e^{-\mathbb{I} (\tmmathbf{\alpha}) /
    \hbar} (1 \otimes e^{\hbar \partial_{P (\varepsilon, \Lambda)}})
    e^{\mathbb{I} (\tmmathbf{\alpha}) / \hbar}\\
    \mathcal{O} (\mathcal{E}) [[\hbar]] \quad &
    \xrightarrow{\LARGE{\mathbb{I}_{\theta_{\Lambda}}}} & \quad \mathcal{O}
    (L') \otimes \mathcal{O} (\mathcal{E}_L) [[\hbar]]
  \end{eqnarray*}
  where $\mathbb{I}_{\theta_{\varepsilon}}, \mathbb{I}_{\theta_{\Lambda}} :
  \mathcal{O} (\mathcal{E}) \rightarrow \mathcal{O} (L') \otimes \mathcal{O}
  (\mathcal{E}_L)$ denote the algebraic isomorphisms induced by
  $\theta_{\varepsilon}, \theta_{\Lambda}$, respectively (details see Theorem
  \ref{ppstnrgrnofspltgmycvtn}).
\end{itemizedot}
As a consistency check, the renormalized splitting interpolates between the
ill-defined ``extension by zero'' in the original work
{\cite{cattaneo2018perturbative}} and the ``bulk to boundary propagator''
known to physicists:
\[ 
   \raisebox{-0.361672163294411\height}{\includegraphics[width=14.4518398268398cm,height=0.919831431195067cm]{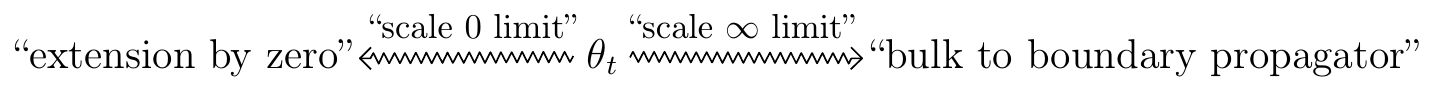}}
\]
(the ``scale $\infty$'' here is not taking naive $t \rightarrow + \infty$ in
our setting, but corresponds to the ``$\Delta^{\infty} $'' in the proof of
{\cite[Proposition 2.3.1]{wangyantqm202203}}).

Above should be the content of homotopic renormalization in BV-BFV formalism
which applies in general.

In this paper, structures purely on boundary (e.g., the BFV operator) do not
need regularization. If the dimension of spacetime is larger than one, the
boundary field space will be typically infinite dimensional, which suggests
that homotopic renormalization should also involve the BFV operator. We leave
this consideration for later study.

\subsubsection*{BV-BFV formalism and the approach within BV formalism}

From the BV-BFV description of TQM we read out
(\ref{bvqmerstctdfldspctqmrplus51}), which is the QME written in
{\cite[Section 3]{wangyantqm202203}}. Moreover, we characterize its generic
solutions in Proposition \ref{thm51bvdscrptntqmrpstvlstplsend5151} based on
discussions for mQME. This suggests that BV-BFV formalism could imply the
approach within BV formalism.

Besides, if (\ref{bvqmerstctdfldspctqmrplus51}) is satisfied, the mQME in
Definition \ref{mqmetqmdfntn36rlylstlbl} can be reinterpreted as a condition
that the map $\mathbb{I}_{(0, t)}$ defined in (\ref{wvfctnisamp56}) is a
cochain map:
\[ \mathbb{I}_{(0, t)} : \left( \mathcal{O} (L) [[\hbar]], \frac{- 1}{\hbar}
   \tmmathbf{\Omega}_L^{\tmop{right}} (-, H^{\partial}) \right) \rightarrow
   (\mathcal{O} (\mathcal{E}_L) [[\hbar]], \mathd + \hbar \partial_{K_t} + \{
   I_t |_{\mathcal{E}_L}, - \}_t), \]
see (\ref{cchnmpeftbdrtoglblobsvbls57}) and its following discussion. This
translation connects the ``wave function'' interpretation of ``state''
inherited in BV-BFV formalism and the ``structure map'' interpretation of
``state'' from factorization algebra perspective in
{\cite{costello_gwilliam_2016,costello_gwilliam_2021}}. It may inspire further
comparative studies between BV-BFV formalism and other frameworks.

\subsubsection*{Configuration space techniques}

The space $\mathbb{R}_{\geqslant 0} [n]$ introduced in Definition
\ref{dfntn31ptlcmpctfctnofcfgrtnspc310} simplifies our analysis of crucial
properties of TQM, including the UV finiteness (Proposition
\ref{tqmuvfntthmprf}) and the BV anomaly (Lemma \ref{bvanmlcmptnn37}). Such a
configuration space technique originates from
{\cite{kontsevich1994feynman,axelrod1994chern,getzler1994operads}}, and
reflects the power of geometric considerations. For QFT's in $2$D, another
geometric renormalization method of regularized integral introduced in
{\cite{rglzditglrmsrfcmdlfmlszj2021}} (see also {\cite{guisili2021elliptic}})
may applies. Maybe this method can be incorporated in BV-BFV formalism as
well.

\subsection{Organization of the paper}

The paper is organized as follows.

In Section \ref{sectn2algprp220307} we briefly introduce perturbative BV
formalism and perturbative BV-BFV formalism, and fix notations on structures
such as Moyal product and Weyl quantization for later use.

Section \ref{sctnbvbfvtqmpstvrllne3} is devoted to TQM on
$\mathbb{R}_{\geqslant 0}$. In Section \ref{cttoffreethytqmsubsectn31guagu},
we formulate homotopic renormalization in BV-BFV formalism for free theory.
Rigorous mQME for interactive theory is stated in Section
\ref{tctntofitctnthytqmpstvrllnsctn32}, followed by a description of its
generic solutions. For TQM on interval, parallel results are obtained in
Section \ref{bvbfvdscptntqmintvlsctn4ttl}.

In Section \ref{sctn5555555lst}, we extract the BV description of TQM in
{\cite{wangyantqm202203}} from the BV-BFV description in Section
\ref{sctnbvbfvtqmpstvrllne3} and Section \ref{bvbfvdscptntqmintvlsctn4ttl}.
Then, mQME is reinterpreted from the perspective of factorization algebra
developed in {\cite{costello_gwilliam_2016,costello_gwilliam_2021}}.

\begin{acknowledgments*}
  We would like to thank Si Li, Nicolai Reshetikhin, Kai Xu, Eugene
  Rabinovich, Philsang Yoo, Brian Williams, Owen Gwilliam, Keyou Zeng for
  illuminating discussion. We especially thank Si Li for invaluable
  conversation and guidance on this work. This work was supported by National
  Key Research and Development Program of China (NO. 2020YFA0713000). Part of
  this work was done in Spring 2022 while M. W. was visiting Center of
  Mathematical Sciences and Applications at Harvard and Perimeter Institute.
  He thanks for their hospitality and provision of excellent working
  enviroment. Research at Perimeter Institute is supported in part by the
  Government of Canada through the Department of Innovation, Science and
  Economic Development Canada and by the Province of Ontario through the
  Ministry of Colleges and Universities.
\end{acknowledgments*}

\begin{convention*}
  {\tmdummy}
  
  \begin{itemizedot}
    \item Let $V$ be a $\mathbb{Z}$-graded $k$-vector space. We use $V_m$ to
    denote its degree $m$ component. Given homogeneous element $a \in V_m$, we
    let $| a | = m$ be its degree.
    \begin{itemizeminus}
      \item $V [n]$ denotes the degree shifting of $V$ such that $V [n]_m =
      V_{n + m}$.
      
      \item $V^{\ast}$ denotes its linear dual such that $V^{\ast}_m =
      \tmop{Hom}_k (V_{- m}, k)$. Our base field $k$ will mainly be
      $\mathbb{R}$.
      
      \item $\tmop{Sym}^m (V)$ and $\wedge^m (V)$ denote the $m$-th power
      graded symmetric product and graded skew-symmetric product respectively.
      We also denote
      \[ \tmop{Sym} (V) \assign \bigoplus_{m \geqslant 0} \tmop{Sym}^m (V),
         \quad \widehat{\tmop{Sym}} (V) \assign \prod_{m \geqslant 0}
         \tmop{Sym}^m (V) . \]
      The latter is a graded symmetric algebra with the former being its
      subalgebra. We will omit the multiplication mark for this product in
      expressions (unless confusion occurs).
      
      \item We call $\mathcal{O} (V) \assign \widehat{\tmop{Sym}} (V^{\ast})$
      the function ring on $V$.
      
      \item $V [[\hbar]], V ((\hbar))$ denote formal power series and Laurent
      series respectively in a variable $\hbar$ valued in $V$.
    \end{itemizeminus}
    \item We use the Einstein summation convention throughout this work.
    
    \item We use $(\pm)_{\tmop{Kos}}$ to represent the sign factors determined
    by Koszul sign rule. We always assume this rule in dealing with graded
    objects.
    \begin{itemizeminus}
      \item Example: let $j$ be a homogeneous linear map on $V$, then
      $j^{\ast}$ denotes the induced linear map on $V^{\ast}$: for $\forall f
      \in V^{\ast}, a \in V$ being homogeneous,
      \[ j^{\ast} f (a) \assign (\pm)_{\tmop{Kos}} f (j (a)) \quad \text{with
         } (\pm)_{\tmop{Kos}} = (- 1)^{| j | | f |} \text{ here.} \]
      \item Example: let $f, g, h \in V^{\ast}$ be homogeneous elements, then
      $f \otimes g \otimes h \in (V^{\ast})^{\otimes 3}$ is regarded as an
      element in $(V^{\otimes 3})^{\ast}$: for $\forall a, b, c \in V$ being
      homogeneous,
      \[ (f \otimes g \otimes h) (a \otimes b \otimes c) \assign
         (\pm)_{\tmop{Kos}} f (a) g (b) h (c) \hspace{1.5em} \text{with }
         (\pm)_{\tmop{Kos}} = (- 1)^{| h | | a | + | h | | b | + | g | | a |}
         \text{ here.} \]
      \item Example: let $(\mathcal{A}, \cdot)$ be a graded algebra, then $[-,
      -]$ means the graded commutator, i.e, for homogeneous elements $a, b$,
      \[ [a, b] \assign a \cdot b - (\pm)_{\tmop{Kos}} b \cdot a \quad
         \text{with } (\pm)_{\tmop{Kos}} = (- 1)^{| a | | b |} \text{ here.}
      \]
    \end{itemizeminus}
    \item We fix an embedding of vector spaces $\tmop{Sym}^m (V)
    \hookrightarrow V^{\otimes m}$ by
    \[ a_1 a_2 \cdots a_m \rightarrow \sum_{\sigma \in \mathbf{S}_m}
       (\pm)_{\tmop{Kos}} a_{\sigma (1)} \otimes a_{\sigma (2)} \otimes \cdots
       \otimes a_{\sigma (m)}, \]
    where $\mathbf{S}_m$ denotes the symmetric group. Accordingly, any $f_1
    f_2 \cdots f_m \in \tmop{Sym}^m (V^{\ast})$ is regarded as an element in
    $(\tmop{Sym}^m (V))^{\ast}$: for $\forall a_1 a_2 \cdots a_m \in
    \tmop{Sym}^m (V)$,
    \[ f_1 f_2 \cdots f_m (a_1 a_2 \cdots a_m) = m! \sum_{\sigma \in
       \mathbf{S}_m} (\pm)_{\tmop{Kos}} f_1 (a_{\sigma (1)}) f_2 (a_{\sigma
       (2)}) \cdots f_m (a_{\sigma (m)}) . \]
    \item We call $(V, d)$ a cochain complex if $d$ is a degree $1$ map on the
    graded vector space $V$ such that $d^2 = 0$. Such $d$ is called a
    differential. A cochain map $f : (V, d) \rightarrow (W, b)$ is a degree
    $0$ map from $V$ to $W$ such that $b f = f d$.
    
    \item Given a manifold $X$, we denote the space of real smooth forms by
    \[ \Omega^{\bullet} (X) = \bigoplus_k \Omega^k (X) \]
    where $\Omega^k (X)$ is the subspace of $k$-forms, lying at degree $k$.
    
    \item Now that we have mentioned differential forms, definitely we will
    work with infinite dimensional functional spaces that carry natural
    topologies. The above notions for $V$ will be generalized as follows. We
    refer the reader to {\cite{treves2006topological}} or {\cite[Appendix
    2]{costellorenormalization}} for further details. Besides, {\cite[Appendix
    A]{rabinovich2021factorization}} contains specialized discussion for
    sections of vector bundles with boundary conditions.
    \begin{itemizeminus}
      \item All topological vector spaces we consider will be nuclear and we
      still use $\otimes$ to denote the completed projective tensor product.
      For example, given two manifolds $X, Y$, we have a canonical isomorphism
      \[ C^{\infty} (X) \otimes C^{\infty} (Y) = C^{\infty} (X \times Y) . \]
      \item In the involved categories, dual space is defined to be the
      continuous linear dual, equipped with the topology of uniform
      convergence of bounded subsets. We still use $(-)^{\ast}$ to denote
      taking such duals.
    \end{itemizeminus}
  \end{itemizedot}
\end{convention*}

\section{Algebraic Preliminaries}\label{sectn2algprp220307}

In this section we collect basics and fix notations on perturbative BV
formalism and perturbative BV-BFV formalism.

\subsection{Perturbative BV quantization}

In BV formalism, the algebraic prototype of a free QFT is the following:

\begin{definition}
  \label{dfntndgbvalgbrasec2229}A {\tmstrong{differential Batalin-Vilkovisky
  (BV) algebra}} is a triple $(\mathcal{A}, Q, \Delta)$ where
  \begin{itemizedot}
    \item $\mathcal{A}$ is a $\mathbb{Z}$-graded commutative associative
    unital algebra. Assume the base field is $\mathbb{R}$.
    
    \item $Q : \mathcal{A} \rightarrow \mathcal{A}$ is a derivation of degree
    $1$ such that $Q^2 = 0$.
    
    \item $\Delta : \mathcal{A} \rightarrow \mathcal{A}$ is a linear operator
    of degree $1$ such that $\Delta^2 = 0$, and $[Q, \Delta] = Q \Delta +
    \Delta Q = 0$. We call $\Delta$ the {\tmstrong{BV operator}}.
    
    \item $\Delta$ is a ``second-order'' operator w.r.t. the product of
    $\mathcal{A}$. Precisely, define the {\tmstrong{BV bracket}} $\{ -, - \} :
    \mathcal{A} \otimes \mathcal{A} \rightarrow \mathcal{A}$ to be the failure
    of $\Delta$ being a derivation:
    \[ \{ a, b \} \assign \Delta (a b) - (\Delta a) b - (- 1)^{| a |} a \Delta
       b, \qquad \text{for } \forall a, b \in \mathcal{A}. \]
    Then for $\forall a \in \mathcal{A}$, $\{ a, - \}$ is a derivation of
    degree $(| a | + 1)$: for $\forall b, c \in \mathcal{A}$
    \[ \{ a, b c \} = \{ a, b \} c + (\pm)_{\tmop{Kos}} b \{ a, c \}, \qquad
       \text{with } (\pm)_{\tmop{Kos}} = (- 1)^{| b | | a | + | b |} \text{
       here} . \]
  \end{itemizedot}
\end{definition}

Let $\hbar$ be a formal variable of degree $0$. We can extend the above $Q,
\Delta$ to $\mathbb{R} [[\hbar]]$-linear operators on $\mathcal{A} [[\hbar]]$.
Then, $(\mathcal{A}, Q, \Delta)$ being a differential BV algebra implies $Q +
\hbar \Delta$ is a differential on $\mathcal{A} [[\hbar]]$. The cochain
complex $(\mathcal{A} [[\hbar]], Q + \hbar \Delta)$ models ``the observables
of the free theory''.

There is a systematic way to twist (i.e., perturb) this complex, sketched in
the following.

\begin{definition}
  Let $(\mathcal{A}, Q, \Delta)$ be a differential BV algebra. A degree $0$
  element $I \in \mathcal{A} [[\hbar]]$ is said to satisfy {\tmstrong{quantum
  master equation (QME)}} if
  \begin{equation}
    Q I + \hbar \Delta I + \frac{1}{2} \{ I, I \} = 0, \label{qmesmplestcase}
  \end{equation}
  or formally,
  \[ (Q + \hbar \Delta) e^{I / \hbar} = 0. \]
\end{definition}

It is direct to check that (\ref{qmesmplestcase}) implies this formal
conjugation of operators on $\mathcal{A} [[\hbar]]$:
\[ Q + \hbar \Delta + \{ I, - \} = e^{- I / \hbar} (Q + \hbar \Delta) e^{I /
   \hbar}, \]
which implies $(Q + \hbar \Delta + \{ I, - \})^2 = 0$. The cochain complex
$(\mathcal{A} [[\hbar]], Q + \hbar \Delta + \{ I, - \})$ models ``the
observables of the interactive theory''. So, roughly speaking, perturbative BV
quantization amounts to find such ``interactive action functional'' $I$ that
solves the QME.

\begin{remark*}
  To ensure $(Q + \hbar \Delta + \{ I, - \})^2 = 0$, the QME
  (\ref{qmesmplestcase}) can be relaxed to the condition that $Q I + \hbar
  \Delta I + \frac{1}{2} \{ I, I \}$ is a central element with respect to the
  BV bracket. We will not take this into consideration here, but it can be
  crucial in certain cases (see e.g., {\cite[Section
  3.3]{2017qtztnalgindexsiliqinli}}).
\end{remark*}

The differential BV algebras in this paper are all function rings on cochain
complexes, with the BV operators being ``contraction with closed rank-$2$
symmetric tensors''. Precisely, let $(\mathcal{E}, Q)$ be a cochain complex,
$K \in \tmop{Sym}^2 (\mathcal{E})$ satisfying $Q K = 0$, $| K | = 1$, then
\[ (\mathcal{O} (\mathcal{E}) \assign \widehat{\tmop{Sym}}
   (\mathcal{E}^{\ast}), Q, \Delta \assign \partial_K) \]
is a differential BV algebra, where
\begin{itemizedot}
  \item we still use $Q$ to denote the derivation on $\mathcal{O}
  (\mathcal{E})$ extended from $Q^{\ast}$ on $\mathcal{E}^{\ast}$;
  
  \item $\partial_K (\tmop{Sym}^{\leqslant 1} (\mathcal{E}^{\ast})) = 0$, and
  for $n \geqslant 2$, $\forall f_1, f_2, \ldots, f_n \in \mathcal{E}^{\ast}$,
  \begin{equation}
    \partial_K (f_1 f_2 \cdots f_n) = \sum_{i < j} (\pm)_{\tmop{Kos}}  (f_i
    f_j (K)) f_1 \ldots \widehat{f_i} \ldots \widehat{f_j} \ldots f_n .
    \label{2tnsrctctndfntnclrfy}
  \end{equation}
\end{itemizedot}
To really use such a differential BV algebra to describe a free QFT, the
choice of $K$ is not arbitrary.

The cochain $(\mathcal{E}, Q)$ models the ``field space'' of the QFT. If the
spacetime manifold is closed, there should be a nondegenerate pairing $\omega
\in \wedge^2 (\mathcal{E}^{\ast}), | \omega | = - 1$ such that $Q \omega = 0$.
$(\mathcal{E}, Q, \omega)$ is called a {\tmstrong{dg $(- 1)$-symplectic vector
space}}, and $K$ should be chosen according to this structure. For example, if
$\mathcal{E}$ is finite dimensional, $\omega$ will induce a vector space
isomorphism $\mathcal{E} \rightarrow \mathcal{E}^{\ast}$ by sending $e \in
\mathcal{E}$ to $\omega (e, -) \in \mathcal{E}^{\ast}$. This further induces a
vector space isomorphism:
\[ \tmop{Sym}^2 (\mathcal{E}) \rightarrow \wedge^2 (\mathcal{E}^{\ast}) . \]
Then, $K \in \tmop{Sym}^2 (\mathcal{E})$ is the preimage of $\frac{1}{2}\omega \in
\wedge^2 (\mathcal{E}^{\ast})$ under this map. For more general field space
$\mathcal{E}$, the construction of $K$ is less straightforward. We will adopt
Costello's homotopic renormalization method {\cite{costellorenormalization}}
to deal with the models appearing in this paper.

\subsection{Perturbative BV-BFV quantization}\label{ptbtvbvbfvsktchmyway}

As mentioned above, we expect the field space of a perturbative QFT on a
closed manifold to be modelled by a dg $(- 1)$-symplectic vector space. If the
spacetime has boundary, this picture needs modification.

In {\cite{cattaneo2014classical,cattaneo2018perturbative}}, Cattaneo, Mnev and
Reshetikhin proposed a candidate modification. Their proposal is called
``BV-BFV formalism'' (BFV for Batalin-Fradkin-Vilkovisky
{\cite{batalin1983generalizedbf0904,fradkin1975quantization0904}}), which
generalizes BV formalism to encompass the presence of spacetime boundary. The
process of using this formalism to perturbatively quantize a QFT is sketched
as follows.

\begin{definition}
  \label{freebvbfvptbtvmydefntn}A {\tmstrong{free BV-BFV pair}} $(\mathcal{E},
  Q, \omega, \mathcal{E}^{\partial}, Q^{\partial}, \omega^{\partial}, \pi)$ is
  the following:
  \begin{itemizedot}
    \item $(\mathcal{E}, Q)$, $(\mathcal{E}^{\partial}, Q^{\partial})$ are
    cochain complexes.
    
    \item $\pi : (\mathcal{E}, Q) \rightarrow (\mathcal{E}^{\partial},
    Q^{\partial})$ is a cochain map, and it is surjective.
    
    \item $\omega^{\partial} \in \wedge^2 ((\mathcal{E}^{\partial})^{\ast}), |
    \omega | = 0$ is a nondegenerate pairing, such that $Q^{\partial}
    \omega^{\partial} = 0$.
    
    \item $\omega \in \wedge^2 (\mathcal{E}^{\ast}), | \omega | = - 1$ is a
    pairing, such that $Q \omega = \pi^{\ast} \omega^{\partial}$.
  \end{itemizedot}
\end{definition}

The physical meaning is that, $\mathcal{E}$ is the ``bulk field space'',
$\mathcal{E}^{\partial}$ is the ``boundary field space'', and $\pi$ is
``restriction to the boundary''. Usually $\omega$ is nondegenerate, then the
condition $Q \omega = \pi^{\ast} \omega^{\partial}$ says that, the failure of
$(\mathcal{E}, Q, \omega)$ being a dg $(- 1)$-symplectic vector space is
controlled by the boundary data $(\mathcal{E}^{\partial}, Q^{\partial},
\omega^{\partial})$.

\begin{remark*}
  We refer the reader to {\cite[Section 7]{cattaneo2020introduction}} for the
  global version of this notion, which is not necessary in the scope of this
  paper.
\end{remark*}

\begin{definition}
  \label{plrztnandspltgmycvtn}Given a free BV-BFV pair $(\mathcal{E}, Q,
  \omega, \mathcal{E}^{\partial}, Q^{\partial}, \omega^{\partial}, \pi)$, a
  {\tmstrong{polarization}} is a Lagrangian decomposition of
  $\mathcal{E}^{\partial}$ compatible with $Q^{\partial}$, namely,
  \begin{equation}
    \mathcal{E}^{\partial} = L \oplus L' \quad \text{such that} \quad
    Q^{\partial} L \subset L, Q^{\partial} L' \subset L', \label{plrztnbvbfv}
  \end{equation}
  and $\omega^{\partial} \in (L^{\ast} \otimes (L')^{\ast}) \oplus
  ((L')^{\ast} \otimes L^{\ast})$ only pairs $L$ with $L'$.
  
  A {\tmstrong{splitting}} compatible with the polarization
  $\mathcal{E}^{\partial} = L \oplus L'$ is a degree $0$ linear map
  \begin{equation}
    \theta : L' \rightarrow \mathcal{E}, \label{cptbsplttng}
  \end{equation}
  such that $\pi \theta : L' \rightarrow \mathcal{E}^{\partial} = L \oplus L'$
  is $(0, \tmop{id}_{L'})$.
\end{definition}

Given a polarization as (\ref{plrztnbvbfv}), let $p_L, p_{L'}$ denote the
projections from $\mathcal{E}^{\partial}$ to $L, L'$, respectively. The
restricted (bulk) field space with boundary condition $L$ is defined as:
\[ \mathcal{E}_L \assign \ker (p_{L'} \pi) \subset \mathcal{E}. \]
Then, $(\mathcal{E}_L, Q)$ is a cochain complex, and $(Q \omega)
|_{\mathcal{E}_L} = 0$. Usually, $\omega$ is nondegenerate on $\mathcal{E}_L$,
hence $(\mathcal{E}_L, Q, \omega)$ is a dg $(- 1)$-symplectic vector space.
Suppose we have chosen a $Q$-closed tensor $K \in \tmop{Sym}^2
(\mathcal{E}_L)$ according to $(\mathcal{E}_L, Q, \omega)$, then $(\mathcal{O}
(\mathcal{E}_L), Q, \partial_K)$ is a differential BV algebra. We could use it
to define the theory still within BV formalism (which is the method in
{\cite{wangyantqm202203}}). However, the spirit of BV-BFV formalism is that,
we should work on $\mathcal{E}$ instead of the subspace $\mathcal{E}_L \subset
\mathcal{E}$ only.

The splitting (\ref{cptbsplttng}) induces an isomorphism $\mathcal{E} \simeq
L' \oplus \mathcal{E}_L$, which further induces an algebraic isomorphism
\begin{equation}
  \mathbb{I}_{\theta} : \mathcal{O} (\mathcal{E}) \rightarrow \mathcal{O} (L')
  \otimes \mathcal{O} (\mathcal{E}_L) . \label{ismbtwnfctnrgonfldspc}
\end{equation}
Then, the consistency condition of a perturbative QFT in BV-BFV formalism is
the following:

\begin{definition*}
  A degree $0$ functional $I \in \mathcal{O} (\mathcal{E}) [[\hbar]]$ is said
  to satisfy the {\tmstrong{modified quantum master equation (mQME)}} if there
  is a certain $\mathbb{R} [[\hbar]]$-linear differential $\tmmathbf{\Omega}$
  on $\mathcal{O} (L') [[\hbar]]$, called the {\tmstrong{BFV operator}}, such
  that
  \[ (1 \otimes (\hbar Q + \hbar^2 \partial_K) +\tmmathbf{\Omega} \otimes 1)
     \mathbb{I}_{\theta} (e^{I / \hbar}) = 0 \]
  for proper choice of $K \in \tmop{Sym}^2 (\mathcal{E}_L)$ and $\theta : L'
  \rightarrow \mathcal{E}$.
\end{definition*}

Here we do not explain the meaning of ``proper choice'' and constraints on
$\tmmathbf{\Omega}$. Roughly speaking, mQME says that the BV anomaly of the
theory defined by $I$ on the (full) bulk field space $\mathcal{E}$ is
cancelled by some boundary data. Precise form of mQME for the concrete models
we study will be stated in later sections.

\subsection{Moyal product and Weyl
quantization}\label{mylprdctwlqtztnalgprlmlry}

We list several standard algebraic structures here for later convenience.

Let
\begin{equation}
  (V = L \oplus L', \omega^{\partial}) \label{gsplctcvcspclgrgndcmp}
\end{equation}
be a finite dimensional graded symplectic vector space endowed with a
Lagrangian decomposition. Namely, $\omega^{\partial} \in \wedge^2 (V^{\ast}),
| \omega | = 0$ is a nondegenerate pairing that only pairs $L$ with $L'$. For
$v \in V$, the map $v \rightarrow \omega^{\partial} (v, -)$ induces a graded
vector space isomorphism $V \simeq V^{\ast}$, hence also $\wedge^2 V \simeq
\wedge^2 (V^{\ast})$. Let
\begin{equation}
  K^{\partial} \in \wedge^2 (V) \label{invrssplctc27}
\end{equation}
be the image of $\omega^{\partial}$ under $\wedge^2 (V^{\ast}) \simeq \wedge^2
V$. Regarded as an element in $V^{\otimes 2}$, we can write
\begin{equation}
  K^{\partial} = K^{\partial}_- + K^{\partial}_+, \quad K^{\partial}_- \in L
  \otimes L', \enspace K^{\partial}_+ \in L' \otimes L, \text{ s.t. } \sigma
  K^{\partial}_- = - K^{\partial}_+, \label{splctckrnldecmpstn}
\end{equation}
where $\sigma$ permutes the two factors of $V^{\otimes 2}$.

\subsubsection*{Moyal product}

There is a graded associative algebra structure $(\mathcal{O} (V) [[\hbar]],
\star_{\hbar})$: for $J, F \in \mathcal{O} (V) [[\hbar]]$,
\begin{equation}
  J \star_{\hbar} F \assign e^{- \hbar \partial_{K^{\partial} /
  2}}_{\tmop{cross}} (J, F), \label{moyalprodmynttn}
\end{equation}
where for $\forall G \in V^{\otimes 2}$ of degree $0$, $\forall j_1, \ldots,
j_m, f_1, \ldots, f_n \in V^{\ast}$,
\begin{eqnarray}
  &  & e^{\partial_G}_{\tmop{cross}} (j_1 j_2 \cdots j_m, f_1 f_2 \cdots f_n)
  \nonumber\\
  & \assign & \sum_{s = 0}^{+ \infty} {\sum_{1 \leqslant l_1 < \cdots < l_s
  \leqslant m}}  \sum_{\tmscript{\begin{array}{c}
    1 \leqslant r_1, \ldots, r_s \leqslant n\\
    r_a \neq r_b \text{ for } a \neq b
  \end{array}}} (\pm)_{\tmop{Kos}}  ((j_{l_1} \otimes f_{r_1}) (G)) \cdots
  ((j_{l_s} \otimes f_{r_s}) (G)) \nonumber\\
  &  & \hspace{13em} \times \left( j_1 \cdots \widehat{j_{l_1}} \cdots
  \widehat{j_{l_s}} \cdots j_m f_1 \cdots \widehat{f_{r_1}} \cdots
  \widehat{f_{r_s}} \cdots f_n \right) .  \label{crossexpcontractnmycvtn}
\end{eqnarray}
$\star_{\hbar}$ is called the {\tmstrong{Moyal product}}.

\subsubsection*{Weyl quantization}

Given the Lagrangian decomposition $V = L \oplus L'$, {\tmstrong{Weyl
quantization}} defines an action:
\[ \tmmathbf{\Omega}_{L'}^{\tmop{left}} (-, -) : \mathcal{O} (V) [[\hbar]]
   \otimes \mathcal{O} (L') [[\hbar]] \rightarrow \mathcal{O} (L') [[\hbar]],
\]
for $J \in \mathcal{O} (V) [[\hbar]], g \in \mathcal{O} (L') [[\hbar]]$,
\begin{equation}
  \tmmathbf{\Omega}_{L'}^{\tmop{left}} (J, g) \assign p_{L'} \left( e^{- \hbar
  \partial_{K_-^{\partial}}}_{\tmop{cross}} \left( e^{\hbar
  \partial_{(K^{\partial}_+ - K^{\partial}_-) / 4}} J, g \right) \right)
  \label{weylqtztnccrtfmlmycnvntn}
\end{equation}
where $p_{L'}$ here denotes the projection $\mathcal{O} (V) \rightarrow
\mathcal{O} (L')$ induced by $V = L \oplus L'$, and $\partial_{(K^{\partial}_+
- K^{\partial}_-) / 4}$ acting on $\mathcal{O} (V)$ is defined in the same way
as (\ref{2tnsrctctndfntnclrfy}). $e^{\hbar \partial_{(K^{\partial}_+ -
K^{\partial}_-) / 4}}$ should be regarded as the transformation from
``symmetric (Weyl) ordering'' to ``normal ordering''. It is direct to verify
that $\tmmathbf{\Omega}_{L'}^{\tmop{left}} (-, -)$ makes $\mathcal{O} (L')
[[\hbar]]$ a graded left module over $(\mathcal{O} (V) [[\hbar]],
\star_{\hbar})$.

Similarly, $\mathcal{O} (L) [[\hbar]]$ has a graded right module structure
over $(\mathcal{O} (V) [[\hbar]], \star_{\hbar})$:
\[ \tmmathbf{\Omega}_L^{\tmop{right}} (-, -) : \mathcal{O} (L) [[\hbar]]
   \otimes \mathcal{O} (V) [[\hbar]] \rightarrow \mathcal{O} (L) [[\hbar]], \]
for $J \in \mathcal{O} (V) [[\hbar]], f \in \mathcal{O} (L) [[\hbar]]$,
\begin{equation}
  \tmmathbf{\Omega}_L^{\tmop{right}} (f, J) \assign p_L \left( e^{- \hbar
  \partial_{K_-^{\partial}}}_{\tmop{cross}} \left( f, e^{\hbar
  \partial_{(K^{\partial}_+ - K^{\partial}_-) / 4}} J \right) \right) .
  \label{weylqtztnrightactnmycnvntn}
\end{equation}
Moreover, we have a nondegenerate pairing
\[ \ll -, - \gg : \mathcal{O} (L) [[\hbar]] \otimes \mathcal{O} (L') [[\hbar]]
   \rightarrow \mathbb{R} [[\hbar]], \]
for $f \in \mathcal{O} (L) [[\hbar]], g \in \mathcal{O} (L') [[\hbar]]$,
\begin{equation}
  \ll f, g \gg \assign p_L p_{L'} \left( e^{- \hbar
  \partial_{K_-^{\partial}}}_{\tmop{cross}} (f, g) \right) .
  \label{hilbertpairingnondegnmycnvntn}
\end{equation}
It is direct to verify
\begin{equation}
  \ll f, p_{L'} \left( e^{- \hbar \partial_{K_-^{\partial}}}_{\tmop{cross}}
  (J, g) \right) \gg = \ll p_L \left( e^{- \hbar
  \partial_{K_-^{\partial}}}_{\tmop{cross}} (f, J) \right), g \gg
  \label{hilbtprcptbwithmdleactn}
\end{equation}
for $\forall J \in \mathcal{O} (V) [[\hbar]]$. So this pairing is compatible
with the module structures.

\subsubsection*{Constraints on the BFV operator}

In Section \ref{ptbtvbvbfvsktchmyway} we mentioned that there should be a BFV
operator $\tmmathbf{\Omega}$ on $\mathcal{O} (L') [[\hbar]]$ appearing in
mQME. Here we describe it more concretely for the case
$\mathcal{E}^{\partial}$ equals to $V$ in (\ref{gsplctcvcspclgrgndcmp}).

It is direct to check that, for an $\mathbb{R} [[\hbar]]$-linear operator
$\tmmathbf{\Omega}$ on $\mathcal{O} (L') [[\hbar]]$, these two statements are
equivalent:
\begin{itemizedot}
  \item $\tmmathbf{\Omega}$ can be expanded as
  \[ \tmmathbf{\Omega}= \sum_{n = 0}^{+ \infty} \hbar^n \tmmathbf{\Omega}_n \]
  where $\tmmathbf{\Omega}_n$ is a differential operator on $\mathcal{O} (L')$
  of order $\leqslant n$.
  
  \item $\exists J \in \mathcal{O} (V) [[\hbar]]$, such that
  $\tmmathbf{\Omega}=\tmmathbf{\Omega}_{L'}^{\tmop{left}} (J, -)$.
\end{itemizedot}
By compatibility between $\tmmathbf{\Omega}_{L'}^{\tmop{left}} (-, -)$ and
$\star_{\hbar}$, for $J \in \mathcal{O} (V) [[\hbar]]$, these two statements
are equivalent:
\begin{itemizedot}
  \item $\tmmathbf{\Omega}_{L'}^{\tmop{left}} (J, -)$ is a differential, i.e.,
  $| \tmmathbf{\Omega}_{L'}^{\tmop{left}} (J, -) | = 1,
  (\tmmathbf{\Omega}_{L'}^{\tmop{left}} (J, -))^2 = 0$.
  
  \item $| J | = 1, J \star_{\hbar} J = 0$.
\end{itemizedot}
In summary, when the data of boundary field space and polarization is given by
(\ref{gsplctcvcspclgrgndcmp}), we demand that the BFV operator on $\mathcal{O}
(L') [[\hbar]]$ should be $\tmmathbf{\Omega}_{L'}^{\tmop{left}} (J, -)$ for
some $J \in \mathcal{O} (V) [[\hbar]]$ satisfying $| J | = 1, J \star_{\hbar}
J = 0$.

\section{BV-BFV Description of TQM on $\mathbb{R}_{\geqslant
0}$}\label{sctnbvbfvtqmpstvrllne3}

We have discussed topological quantum mechanics (TQM) on
$\mathbb{R}_{\geqslant 0}$ in {\cite{wangyantqm202203}} with Costello's
homotopic renormalization incorporated, based on
{\cite{rabinovich2021factorization}}. There, we only work on ``restricted bulk
field space'', so that the theory can be defined within BV formalism. In this
section, we will refine that discussion using BV-BFV formalism. Moreover, we
will find generic solutions to the mQME for TQM.

\subsection{The content of free theory}\label{cttoffreethytqmsubsectn31guagu}

\subsubsection*{Field spaces}

Given $(V = L \oplus L', \omega^{\partial})$ as (\ref{gsplctcvcspclgrgndcmp}),
let
\begin{equation}
  (\mathcal{E}^{\partial} \assign V, Q^{\partial} \assign 0,
  \omega^{\partial}) \label{freebdrdtmytqmeg}
\end{equation}
be the boundary field data (now the spacetime is $\mathbb{R}_{\geqslant 0}$,
with boundary being $\{ 0 \}$),
\begin{equation}
  \left( \mathcal{E} \assign \Omega^{\bullet} (\mathbb{R}_{\geqslant 0})
  \otimes V, Q \assign \mathd, \omega \assign \int \omega^{\partial} \in
  \wedge^2 (\mathcal{E}_c^{\ast}) \right) \label{freeblkdtbvbfvmytqmeg}
\end{equation}
be the bulk field data, where $\mathd$ is the de Rham differential on
$\Omega^{\bullet} (\mathbb{R}_{\geqslant 0})$. Note that $\omega$ is
{\tmstrong{not}} a pairing on $\mathcal{E}$ now (since integration might be
divergent on $\mathbb{R}_{\geqslant 0}$). Instead, it is a pairing on the
space of compactly supported forms $\mathcal{E}_c \assign \Omega^{\bullet}_c
(\mathbb{R}_{\geqslant 0}) \otimes V \subset \mathcal{E}$. Let
\begin{equation}
  \pi : \mathcal{E} \rightarrow \mathcal{E}^{\partial}
  \label{freebvbfvdtrstctnmpmytqmeg}
\end{equation}
be the pullback of differential forms induced by $\{ 0 \} \hookrightarrow
\mathbb{R}_{\geqslant 0}$. Then, we can verify $(\mathcal{E}_c, \mathd,
\omega, \mathcal{E}^{\partial}, 0, \omega^{\partial}, \pi)$ is a free BV-BFV
pair as in Definition \ref{freebvbfvptbtvmydefntn}. The noncompactness of
$\mathbb{R}_{\geqslant 0}$ makes $\mathcal{E}_c$ rather than $\mathcal{E}$ fit
into this definition, and it also leads to several other technical subtleties
in the interactive theory (see Remark \ref{rmkncmpctpprsrptfctsitactvthy}). We
will not expand on this aspect of the story in this paper and just refer the
reader to {\cite[Section 3]{wangyantqm202203}} for relevant facts.

\subsubsection*{Renormalized BV structure with restricted bulk field space}

The polarization $\mathcal{E}^{\partial} = L \oplus L'$ allows us to define
\begin{equation}
  \mathcal{E}_L \assign \{ f \in \Omega^{\bullet} (\mathbb{R}_{\geqslant 0})
  \otimes V| \pi (f) \in L \}, \label{rstctdfldspctqmrpstvsctn3fml4}
\end{equation}
which is the field space we have studied in {\cite{wangyantqm202203}}. By
constructions in {\cite{rabinovich2021factorization}}, we have a family of
differential BV structures on $\mathcal{O} (\mathcal{E}_L)$, concluded in
{\cite[Proposition 3.0.2]{wangyantqm202203}}. Here we briefly describe it as
follows.

For $\forall t > 0$, let
\begin{eqnarray*}
  H_t & \assign & \frac{1}{\sqrt{4 \pi t}} \left( e^{- \frac{(x - y)^2}{4 t}}
  (\tmop{dy} - \tmop{dx}) - e^{- \frac{(x + y)^2}{4 t}} (\tmop{dy} +
  \tmop{dx}) \right),\\
  \widetilde{H_t} & \assign & \frac{1}{\sqrt{4 \pi t}} \left( \phi (x - y)
  e^{- \frac{(x - y)^2}{4 t}} (\tmop{dy} - \tmop{dx}) - \phi (x + y) e^{-
  \frac{(x + y)^2}{4 t}} (\tmop{dy} + \tmop{dx}) \right)
\end{eqnarray*}
denote two smooth forms on $\mathbb{R}_{\geqslant 0} \times
\mathbb{R}_{\geqslant 0}$, with $\phi \in C^{\infty} (\mathbb{R})$ being a
compactly supported even function which evaluates to $1$ in a neighborhood of
$\{ 0 \}$. Define the {\tmstrong{BV kernel at scale $t$}} to be
\begin{eqnarray}
  K_t & \assign & \frac{1}{2} \left( H_t - \frac{1}{2}  \left( 1 \otimes
  \text{d} + \text{d} \otimes 1 \right) \left( \text{d}^{\tmop{GF}} \otimes 1
  + 1 \otimes \mathd^{\tmop{GF}} \right) \int_0^t \tmop{ds} (\widetilde{H_s} -
  H_s) \right) \otimes K^{\partial}_+ \nonumber\\
  &  & - \frac{1}{2} \sigma \left( H_t - \frac{1}{2}  \left( 1 \otimes
  \text{d} + \text{d} \otimes 1 \right) \left( \text{d}^{\tmop{GF}} \otimes 1
  + 1 \otimes \mathd^{\tmop{GF}} \right) \int_0^t \tmop{ds} (\widetilde{H_s} -
  H_s) \right) \otimes K^{\partial}_-  \label{bvknldfntnbvbfvwrk}
\end{eqnarray}
where $K^{\partial}_+, K^{\partial}_-$ are defined in
(\ref{splctckrnldecmpstn}), $\sigma$ here permutes variables $x$ and $y$, and
$\text{d}^{\tmop{GF}}$ is the Hodge dual to $\mathd$ induced by metric
$\langle \partial_x, \partial_x \rangle = 1$ on $\mathbb{R}_{\geqslant 0}$:
\[ \text{d}^{\tmop{GF}} (f \tmop{dx}) = - \partial_x f \quad \text{for } f \in
   \Omega^0 (\mathbb{R}_{\geqslant 0}) . \]
Then, by direct computation we can verify
\[ K_t \in \tmop{Sym}^2 (\mathcal{E}_L), \quad | K_t | = 1, \mathd K_t = 0. \]
So,
\begin{equation}
  (\mathcal{O} (\mathcal{E}_L), \mathd, \partial_{K_t})
  \label{rnmlzdbvel36fmllstlbl}
\end{equation}
is a differential BV algebra for $\forall t > 0$. Moreover, for $\forall
\varepsilon, \Lambda > 0$, define the {\tmstrong{propagator from scale
$\varepsilon$ to scale $\Lambda$}} to be
\begin{eqnarray}
  P (\varepsilon, \Lambda) & \assign & \left( \frac{- 1}{4}  \left(
  \text{d}^{\tmop{GF}} \otimes 1 + 1 \otimes \mathd^{\tmop{GF}} \right)
  \int_{\varepsilon}^{\Lambda} \tmop{dt} \widetilde{H_t} \right) \otimes
  K^{\partial}_+ \nonumber\\
  &  & + \sigma \left( \frac{1}{4}  \left( \text{d}^{\tmop{GF}} \otimes 1 + 1
  \otimes \mathd^{\tmop{GF}} \right) \int_{\varepsilon}^{\Lambda} \tmop{dt}
  \widetilde{H_t} \right) \otimes K^{\partial}_- . 
  \label{ppgtdfntnbvbfvwrkmycvtn}
\end{eqnarray}
It satisfies
\begin{equation}
  P (\varepsilon, \Lambda) \in \tmop{Sym}^2 (\mathcal{E}_L), \quad K_{\Lambda}
  = K_{\varepsilon} + \mathd P (\varepsilon, \Lambda),
  \label{dppgtisdiffofbvknl}
\end{equation}
hence leading to a conjugation of cochain complexes
\begin{equation}
  (\mathcal{O} (\mathcal{E}_L) [[\hbar]], \mathd + \hbar
  \partial_{K_{\varepsilon}}) \begin{array}{c}
    e^{\hbar \partial_{P (\varepsilon, \Lambda)}}\\
    \rightleftharpoons\\
    e^{- \hbar \partial_{P (\varepsilon, \Lambda)}}
  \end{array} (\mathcal{O} (\mathcal{E}_L) [[\hbar]], \mathd + \hbar
  \partial_{K_{\Lambda}}) . \label{cnjgtnfreebvalgtqmrstctfldspc}
\end{equation}
This relation exhibits Costello's homotopic renormalization (for the free
theory) within BV formalism.

For later discussion we further present several properties of the propagator.

\begin{definition}
  \label{dfntn31ptlcmpctfctnofcfgrtnspc310}For $n \geqslant 1$, define
  $\mathbb{R}_{\geqslant 0} [n]$ to be
  \begin{equation}
    \mathbb{R}_{\geqslant 0} [n] \assign \sqcup_{\sigma \in \mathbf{S}_n} \{
    (x_1, \ldots, x_n) |0 \leqslant x_{\sigma (1)} \leqslant \cdots \leqslant
    x_{\sigma (n)} \} . \label{cnfgspcmyversnpts}
  \end{equation}
  Namely, $\mathbb{R}_{\geqslant 0} [n]$ is the {\tmstrong{disjoint union}} of
  $n!$ connected components. Note that there are $n!$ {\tmstrong{different}}
  copies of $(0, 0, \ldots, 0)$ in $\mathbb{R}_{\geqslant 0} [n]$, labelled by
  each $\sigma \in \mathbf{S}_n$.
\end{definition}

We can regard $\mathbb{R}_{\geqslant 0} [n]$ as a ``partial compactification''
of configuration space of $n$ pairwise different points in
$\mathbb{R}_{\geqslant 0}$. Particularly, we can write
\begin{equation}
  \mathbb{R}_{\geqslant 0} [2] = C_1 \sqcup C_2, \quad \text{where } C_1
  \assign \{ (x, y) |0 \leqslant x \leqslant y \}, C_2 \assign \{ (x, y) | 0
  \leqslant y \leqslant x \} . \label{2ptscfgrtnspccvtn}
\end{equation}
Let $\mathD_2$ denote the diagonal in $\mathbb{R}_{\geqslant 0} \times
\mathbb{R}_{\geqslant 0}$, then there is a canonical embedding
$(\mathbb{R}_{\geqslant 0} \times \mathbb{R}_{\geqslant 0}) \backslash
\mathD_2 \hookrightarrow \mathbb{R}_{\geqslant 0} [2]$, where the image of the
former contains the interior of the latter. By direct check we obtain:

\begin{proposition}
  \label{extddppgtrltnwithlmtofppgtmycvtn}On $(\mathbb{R}_{\geqslant 0} \times
  \mathbb{R}_{\geqslant 0}) \backslash \mathD_2$, the expression
  (\ref{ppgtdfntnbvbfvwrkmycvtn}) is well defined for $\varepsilon = 0,
  \Lambda > 0$. The resulting $P (0, \Lambda)$ is a smooth function on
  $(\mathbb{R}_{\geqslant 0} \times \mathbb{R}_{\geqslant 0}) \backslash
  \mathD_2$, and can be smoothly extended to $\mathbb{R}_{\geqslant 0} [2]$.
  We denote this extension by $\bar{P} (0, \Lambda)$.
  
  Given any $(x, y) \in \mathbb{R}_{\geqslant 0} \times \mathbb{R}_{\geqslant
  0}$, the limit
  \[ \lim_{\varepsilon \rightarrow 0} P (\varepsilon, \Lambda) (x, y) \]
  exists, and
  \begin{eqnarray*}
    P (\varepsilon, \Lambda) (0, 0) & = & 0 ;\\
    \lim_{\varepsilon \rightarrow 0} P (\varepsilon, \Lambda) (x, y) & = &
    \bar{P} (0, \Lambda) (x, y) \quad \text{for } x \neq y ;\\
    \lim_{\varepsilon \rightarrow 0} P (\varepsilon, \Lambda) (x, x) & = &
    (\bar{P} (0, \Lambda) |_{C_1} (x, x) + \bar{P} (0, \Lambda) |_{C_2} (x,
    x)) / 2 \quad \text{for } x > 0,
  \end{eqnarray*}
  where $C_1, C_2$ are defined in (\ref{2ptscfgrtnspccvtn}).
\end{proposition}

It is straightforward to verify:
\begin{equation}
  \bar{P} (0, \Lambda) |_{C_1} (x, x) - \bar{P} (0, \Lambda) |_{C_2} (x, x) =
  - K^{\partial} / 2 \quad \text{for } x \geqslant 0,
  \label{bdrdiffofextddppgt}
\end{equation}
and
\begin{equation}
  \bar{P} (0, \Lambda) |_{C_1} (0, 0) = - K^{\partial}_- / 2, \qquad \bar{P}
  (0, \Lambda) |_{C_2} (0, 0) = K^{\partial}_+ / 2.
  \label{spclbdrvleofextddppgt}
\end{equation}
We can also restrict the relation $K_{\Lambda} = K_{\varepsilon} + \mathd P
(\varepsilon, \Lambda)$ in (\ref{dppgtisdiffofbvknl}) to
$(\mathbb{R}_{\geqslant 0} \times \mathbb{R}_{\geqslant 0}) \backslash
\mathD_2$. Then its $\varepsilon \rightarrow 0$ limit reads
\begin{equation}
  K_{\Lambda} = \mathd \bar{P} (0, \Lambda)
  \label{bvknrlisdextddppgtonconfgspc}
\end{equation}
as a relation between smooth forms. Since $(\mathbb{R}_{\geqslant 0} \times
\mathbb{R}_{\geqslant 0}) \backslash \mathD_2$ contains the interior of
$\mathbb{R}_{\geqslant 0} [2]$, (\ref{bvknrlisdextddppgtonconfgspc}) holds on
the entire $\mathbb{R}_{\geqslant 0} [2]$, where $K_{\Lambda}$ is regarded as
a form on $\mathbb{R}_{\geqslant 0} [2]$ by pullback of the canonical
projection $\mathbb{R}_{\geqslant 0} [2] \rightarrow \mathbb{R}_{\geqslant 0}
\times \mathbb{R}_{\geqslant 0}$.

\begin{remark*}
  The idea of using (compactified) configuration space to analyse computations
  in QFT follows
  {\cite{kontsevich1994feynman,axelrod1994chern,getzler1994operads}}, while
  our Definition \ref{dfntn31ptlcmpctfctnofcfgrtnspc310} is specialized for
  TQM on $\mathbb{R}_{\geqslant 0}$.
\end{remark*}

\subsubsection*{Renormalized splitting}

Now we come to BV-BFV formalism and consider structures on $\mathcal{O}
(\mathcal{E})$. Since $(\mathcal{E}_L, \mathd)$ is a subcomplex of
$(\mathcal{E}, \mathd)$, $(\mathcal{O} (\mathcal{E}), \mathd, \partial_{K_t})$
is also a differential BV algebra, and there is also a conjugation
\begin{equation}
  (\mathcal{O} (\mathcal{E}) [[\hbar]], \mathd + \hbar
  \partial_{K_{\varepsilon}}) \begin{array}{c}
    e^{\hbar \partial_{P (\varepsilon, \Lambda)}}\\
    \rightleftharpoons\\
    e^{- \hbar \partial_{P (\varepsilon, \Lambda)}}
  \end{array} (\mathcal{O} (\mathcal{E}) [[\hbar]], \mathd + \hbar
  \partial_{K_{\Lambda}}) . \label{bvbfvtotalfldspchmtpyrgflowfree}
\end{equation}
However, homotopic renormalization in BV-BFV formalism has additional content,
involving the splitting and the BFV operator. For the current model, the BFV
operator does not need regularization because the boundary field space
$\mathcal{E}^{\partial}$ is finite dimensional. We focus on the splitting in
the following.

For $\forall J \in \mathcal{O} (V) \simeq \mathcal{O} (L') \otimes \mathcal{O}
(L)$, define $\mathbb{I} (J) \in \mathcal{O} (L') \otimes \mathcal{O}
(\mathcal{E}_L)$ to be
\begin{equation}
  \mathbb{I} (J) \assign (1 \otimes (\pi |_{\mathcal{E}_L})^{\ast}) (J),
  \label{nttnspcliso316}
\end{equation}
where $\pi |_{\mathcal{E}_L} : \mathcal{E}_L \rightarrow L$ is the restriction
of (\ref{freebvbfvdtrstctnmpmytqmeg}) to $\mathcal{E}_L \subset \mathcal{E}$.
It is direct to see that for any splitting $\theta$ as in Definition
\ref{plrztnandspltgmycvtn}, $\mathbb{I}_{\theta} (\pi^{\ast} (J)) =\mathbb{I}
(J)$ where $\mathbb{I}_{\theta}$ is the algebraic isomorphism
(\ref{ismbtwnfctnrgonfldspc}).

Let $\tmmathbf{\alpha} \in (L')^{\ast} \otimes L^{\ast}$ be determined by
\begin{equation}
  \tmmathbf{\alpha} (l' \otimes l) \assign \omega^{\partial} (l', l) \text{
  for } l' \in L', l \in L. \label{spclbdrfctnlmycnvtn}
\end{equation}
$\tmmathbf{\alpha}$ can be regarded as an element in $\tmop{Sym}^2
(V^{\ast})$.\footnote{So, $\pi^{\ast} (\tmmathbf{\alpha}) \in \tmop{Sym}^2
(\mathcal{E}^{\ast})$ looks like ``$q (0) p (0)$'' in Darboux coordinates on
boundary field space.} Then,
\[ \mathbb{I} (\tmmathbf{\alpha}) \in (L')^{\ast} \otimes
   (\mathcal{E}_L)^{\ast}, \quad \mathbb{I} (\tmmathbf{\alpha}) (l' \otimes f)
   = \omega^{\partial} (l', \pi (f)) . \]
Recall that for $\forall H \in \mathcal{O} (V) [[\hbar]]$ there is an operator
$\tmmathbf{\Omega}_{L'}^{\tmop{left}} (H, -)$ on $\mathcal{O} (L') [[\hbar]]$
defined in (\ref{weylqtztnccrtfmlmycnvntn}). Then it is direct to see that for
$\forall J \in \mathcal{O} (V) [[\hbar]]$,
\begin{equation}
  e^{-\mathbb{I} (\tmmathbf{\alpha}) / \hbar}
  (\tmmathbf{\Omega}_{L'}^{\tmop{left}} (H, -) \otimes 1) \mathbb{I}
  (e^{\tmmathbf{\alpha}/ \hbar} J) =\mathbb{I} \left( e^{- \hbar
  \partial_{K_-^{\partial}}}_{\tmop{cross}} \left( e^{\hbar
  \partial_{(K^{\partial}_+ - K^{\partial}_-) / 4}} H, J \right) \right)
  \label{usefulfactabta2mycvtn1}
\end{equation}
Now we state the choice of splitting for our model. By
(\ref{spclbdrvleofextddppgt}) we have:

\begin{proposition}
  \label{scltspltg31532ppstn}The map $\theta_t : L' \rightarrow \mathcal{E}$
  determined by
  \begin{equation}
    \theta_t (l') (x) \assign 2 (\tmmathbf{\alpha} (l' \otimes -) \otimes 1)
    (\bar{P} (0, t) |_{C_1} (0, x)) \quad \text{for } \forall l' \in L', x
    \geqslant 0 \label{spclspltscltmycnvtn}
  \end{equation}
  is a splitting in the sense of Definition \ref{plrztnandspltgmycvtn}.
  Namely, $\theta_t (l')$ is a $V$-valued $0$-form on $\mathbb{R}_{\geqslant
  0}$, with value at $x$ being the result of contracting $2\tmmathbf{\alpha}
  (l' \otimes -)$ with the first tensor factor of $\bar{P} (0, t) |_{C_1} (0,
  x) \in V \otimes V$.
\end{proposition}

\begin{definition}
  \label{dfntn34rnzdsplttglstlt}We define the {\tmstrong{renormalized
  splitting}} {\tmstrong{(at scale $t$)}} for the current TQM on
  $\mathbb{R}_{\geqslant 0}$ to be the map $\theta_t$ described in
  (\ref{spclspltscltmycnvtn}). We can rewrite it as
  \[ \theta_t (l') = 2 (\mathbb{I} (\tmmathbf{\alpha}) (l' \otimes -) \otimes
     1) \bar{P} (0, t) |_{C_1} \]
  which makes sense although $\bar{P} (0, t) |_{C_1}$ is not in $\mathcal{E}_L
  \otimes \mathcal{E}$.
\end{definition}

Then, for this $\theta_t$ we can write $\mathbb{I}_{\theta_t}$ as
\begin{equation}
  \mathbb{I}_{\theta_t} (f) = e^{-\mathbb{I} (\tmmathbf{\alpha}) / \hbar}
  \left( e^{\hbar \partial_{2 \bar{P} (0, t) |_{C_1}}}_{\tmop{cross}}
  (e^{\mathbb{I} (\tmmathbf{\alpha}) / \hbar}, f) \right)_{\mathcal{E}_L}
  \quad \text{for } f \in \mathcal{O} (\mathcal{E}) \label{spltgalgismmycvtn}
\end{equation}
where $e^{\hbar \partial_{2 \bar{P} (0, t) |_{C_1}}}_{\tmop{cross}} (-, -)$ is
defined in the same way as (\ref{crossexpcontractnmycvtn}), and the
restriction of functional $(-)_{\mathcal{E}_L}$ only acts on factors coming
from $f$. Note that by definition the first factor of $\mathbb{I}
(\tmmathbf{\alpha})$ does not contract with $2 \bar{P} (0, t) |_{C_1}$ in
(\ref{spltgalgismmycvtn}).

By (\ref{bvknrlisdextddppgtonconfgspc}) we can verify:
\begin{equation}
  (1 \otimes \mathd + \{ \mathbb{I} (\tmmathbf{\alpha}), - \}_t)
  \mathbb{I}_{\theta_t} =\mathbb{I}_{\theta_t} \mathd
  \label{cptbltycdtnofspltgandbvdtmyycccvtn}
\end{equation}
as maps from $\mathcal{O} (\mathcal{E})$ to $\mathcal{O} (L') \otimes
\mathcal{O} (\mathcal{E}_L)$, where $\{ -, - \}_t$ is the BV bracket of
$\partial_{K_t}$ on $\mathcal{O} (\mathcal{E}_L)$. We may regard this relation
as certain compatibility between $\theta_t$ and the differential BV algebra
structure associated to $\mathd$ and $K_t$.

The {\tmstrong{homotopic renormalization of splitting}} in our model reads:

\begin{theorem}
  \label{ppstnrgrnofspltgmycvtn}For $\forall \varepsilon, \Lambda > 0$, the
  renormalized splittings $\theta_{\varepsilon}, \theta_{\Lambda}$ in
  Definition \ref{dfntn34rnzdsplttglstlt} and the conjugation map $e^{\hbar
  \partial_{P (\varepsilon, \Lambda)}}$ in
  (\ref{cnjgtnfreebvalgtqmrstctfldspc}) and
  (\ref{bvbfvtotalfldspchmtpyrgflowfree}) make this diagram commute:
  \begin{eqnarray*}
    \mathcal{O} (\mathcal{E}) [[\hbar]] \quad &
    \xrightarrow{\LARGE{\mathbb{I}_{\theta_{\varepsilon}}}} & \quad
    \mathcal{O} (L') \otimes \mathcal{O} (\mathcal{E}_L) [[\hbar]]\\
    \downarrow \enspace e^{\hbar \partial_{P (\varepsilon, \Lambda)}} &  &
    \hspace{3em} \downarrow \enspace e^{-\mathbb{I} (\tmmathbf{\alpha}) /
    \hbar} (1 \otimes e^{\hbar \partial_{P (\varepsilon, \Lambda)}})
    e^{\mathbb{I} (\tmmathbf{\alpha}) / \hbar}\\
    \mathcal{O} (\mathcal{E}) [[\hbar]] \quad &
    \xrightarrow{\LARGE{\mathbb{I}_{\theta_{\Lambda}}}} & \quad \mathcal{O}
    (L') \otimes \mathcal{O} (\mathcal{E}_L) [[\hbar]]
  \end{eqnarray*}
\end{theorem}

\begin{proof}
  Since $P (\varepsilon, \Lambda) (0, 0) = 0$, $(1 \otimes \partial_{P
  (\varepsilon, \Lambda)}) e^{\mathbb{I} (\tmmathbf{\alpha}) / \hbar} = 0$.
  Then, for $f \in \mathcal{O} (\mathcal{E}) [[\hbar]]$,
  \begin{eqnarray*}
    &  & e^{-\mathbb{I} (\tmmathbf{\alpha}) / \hbar} (1 \otimes e^{\hbar
    \partial_{P (\varepsilon, \Lambda)}}) (e^{\mathbb{I} (\tmmathbf{\alpha}) /
    \hbar} \mathbb{I}_{\theta_{\varepsilon}} (f))\\
    & = & e^{-\mathbb{I} (\tmmathbf{\alpha}) / \hbar} (1 \otimes e^{\hbar
    \partial_{P (\varepsilon, \Lambda)}}) \left( e^{\hbar \partial_{2 \bar{P}
    (0, \varepsilon) |_{C_1}}}_{\tmop{cross}} (e^{\mathbb{I}
    (\tmmathbf{\alpha}) / \hbar}, f) \right)_{\mathcal{E}_L}\\
    & = & e^{-\mathbb{I} (\tmmathbf{\alpha}) / \hbar} \left( e^{\hbar \left(
    \partial_{2 \bar{P} (0, \varepsilon) |_{C_1}} + \partial_{2 P
    (\varepsilon, \Lambda)} \right)}_{\tmop{cross}} (e^{\mathbb{I}
    (\tmmathbf{\alpha}) / \hbar}, e^{\hbar \partial_{P (\varepsilon,
    \Lambda)}} f) \right)_{\mathcal{E}_L}\\
    & = & e^{-\mathbb{I} (\tmmathbf{\alpha}) / \hbar} \left( e^{\hbar
    \partial_{2 \bar{P} (0, \Lambda) |_{C_1}}}_{\tmop{cross}} (e^{\mathbb{I}
    (\tmmathbf{\alpha}) / \hbar}, e^{\hbar \partial_{P (\varepsilon,
    \Lambda)}} f) \right)_{\mathcal{E}_L}\\
    & = & \mathbb{I}_{\theta_{\Lambda}} (e^{\hbar \partial_{P (\varepsilon,
    \Lambda)}} f),
  \end{eqnarray*}
  where we have used the fact $ \bar{P} (0, \varepsilon) + P (\varepsilon,
  \Lambda) = \bar{P} (0, \Lambda)$ on $\mathbb{R}_{\geqslant 0} [2]$.
\end{proof}

Theorem \ref{ppstnrgrnofspltgmycvtn}, together with
(\ref{cnjgtnfreebvalgtqmrstctfldspc}) and
(\ref{bvbfvtotalfldspchmtpyrgflowfree}), exhibits the homotopic
renormalization for free TQM in BV-BFV formalism.

\subsection{The content of interactive
theory}\label{tctntofitctnthytqmpstvrllnsctn32}

\subsubsection*{The precise form of mQME}

Now we consider the interactive theory. Let
\[ I^{\partial}, J^{\partial} \in \prod_{n \geqslant 2} \tmop{Sym}^n
   (V^{\ast}) \oplus \hbar \prod_{n > 0} \tmop{Sym}^n (V^{\ast}) [[\hbar]], \]
with degree $| I^{\partial} | = 1, | J^{\partial} | = 0$. Then we define the
``scale $0$'' interaction to be
\begin{equation}
  I_0 \in \mathcal{O}^{> 0} (\mathcal{E}_c) [[\hbar]], \quad I_0 \assign
  \pi^{\ast} (J^{\partial}) + \int_{\mathbb{R}_{\geqslant 0}} I^{\partial}
  \label{scl0itactnblkbdrtrmmycnvntn}
\end{equation}
where $\mathcal{O}^{> 0} (\mathcal{E}_c)$ denotes $\prod_{n > 0} \tmop{Sym}^n
((\mathcal{E}_c)^{\ast})$, with $\mathcal{E}_c \subset \mathcal{E}$ the space
of compactly supported forms. Here, $\pi^{\ast} (J^{\partial})$ is a boundary
term, and $\int_{\mathbb{R}_{\geqslant 0}} I^{\partial}$ denotes the
$\Omega^{\bullet} (\mathbb{R}_{\geqslant 0})$-linear extension of
$I^{\partial}$ followed by integration over $\mathbb{R}_{\geqslant 0}$.

The interaction at scale $t > 0$ should be a functional $I_t \in
\mathcal{O}^{> 0} (\mathcal{E}_c) [[\hbar]]$ generated from $I_0$ by homotopic
renormalization. Precisely, we have:

\begin{proposition}
  \label{tqmuvfntthmprf}``TQM is UV finite''
  
  For $I_0$ defined in (\ref{scl0itactnblkbdrtrmmycnvntn}), the limit
  \begin{equation}
    I_t \assign \lim_{\varepsilon \rightarrow 0} \hbar \log (e^{\hbar
    \partial_{P (\varepsilon, t)}} e^{I_0 / \hbar})
    \label{intactnatfnttscldfntnmm}
  \end{equation}
  is a well-defined element in $\mathcal{O}^{> 0} (\mathcal{E}_c) [[\hbar]]$.
  We call it the ``scale $t$ interaction''.
\end{proposition}

\begin{remark}
  \label{rmkncmpctpprsrptfctsitactvthy}Actually, being a ``local functional'',
  $I_0$ has ``proper support''. The BV kernel $K_t$ and propagator $P
  (\varepsilon, \Lambda)$ also have ``proper supports''. These facts make
  \begin{equation}
    \log (e^{\hbar \partial_{P (\varepsilon, t)}} e^{I_0 / \hbar})
    \label{tempoexprtninuvfntthm}
  \end{equation}
  well-defined, although $P (\varepsilon, t)$ is not compactly supported.
  Moreover, (\ref{tempoexprtninuvfntthm}) has ``proper support'' as well.
  Later calculations involving $\partial_{K_t}$ and its BV bracket also make
  sense by the same reasoning. We refer the reader to {\cite[Section
  3]{wangyantqm202203}} and references therein for relevant explanation.
  Clarification on this point will be omitted in the rest of this paper.
\end{remark}

\begin{proof}
  By standard arguments, the real content of the formula
  (\ref{tempoexprtninuvfntthm}) is a summation over connected Feynman graphs.
  It is direct to check that at any given $\hbar$-order and symmetric tensor
  order, only finite many terms contributes,\footnote{This uses the fact that
  $P (\varepsilon, \Lambda)$ is $0$-form with $P (\varepsilon, \Lambda) (0, 0)
  = 0$, so each bulk vertice will have at least one externel leg to pick up
  $1$-form, and two boundary vertices cannot be directly connected by a
  propagator.} hence (\ref{tempoexprtninuvfntthm}) is well defined.
  
  Consider a term contracting with $\varphi_1, \ldots, \varphi_4 \in
  \mathcal{E}_c$, depicted by the graph below (numeric coefficients are
  omitted):
  \[ 
     \raisebox{-0.802621982752905\height}{\includegraphics[width=14.3888888888889cm,height=3.70984192575102cm]{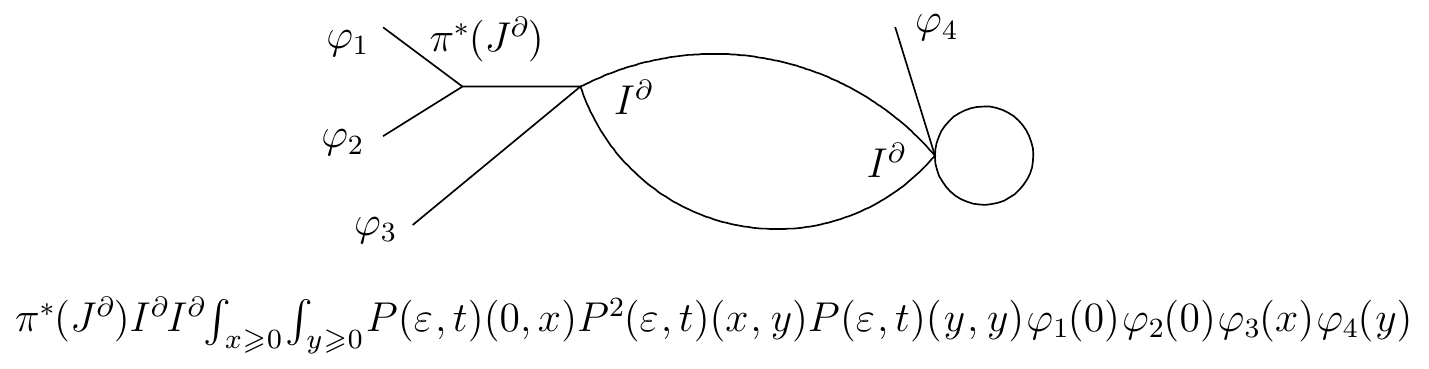}}
  \]
  From this example, it is clear that each term in
  (\ref{tempoexprtninuvfntthm}) is always an integration over
  $\mathbb{R}_{\geqslant 0}^{| V_{\tmop{bulk}} |}$ with the integrand
  constructed from propagators and input fields, where $| V_{\tmop{bulk}} |$
  is the number of ``bulk vertices'' corresponding to $I^{\partial}$. By
  Definition \ref{dfntn31ptlcmpctfctnofcfgrtnspc310}, we have a canonical
  projection $\mathbb{R}_{\geqslant 0} [n] \rightarrow \mathbb{R}_{\geqslant
  0}^n$. So, each term in (\ref{tempoexprtninuvfntthm}) can be expressed as
  integration over $\mathbb{R}_{\geqslant 0} [| V_{\tmop{bulk}} |]$.
  
  By Proposition \ref{extddppgtrltnwithlmtofppgtmycvtn}, while taking the
  $\varepsilon \rightarrow 0$ limit of a certain propagator $P (\varepsilon,
  t)$ in a graph formula, we can replace this $P (\varepsilon, t)$ with a
  linear combination of $\bar{P} (0, t)$. (The type of vertices at the
  endpoints of this propagator will tell us the choice of replacement.) So,
  the $\varepsilon \rightarrow 0$ limit of each term in
  (\ref{tempoexprtninuvfntthm}) is an integration of smooth form over a
  manifold with corners, namely the $I_t$ in (\ref{intactnatfnttscldfntnmm})
  is well defined.
\end{proof}

By definition, for $\varepsilon, \Lambda > 0$,
\begin{equation}
  e^{I_{\Lambda} / \hbar} = e^{\hbar \partial_{P (\varepsilon, \Lambda)}}
  e^{I_{\varepsilon} / \hbar} . \label{RGflowofrnmlzditactn321}
\end{equation}
Now we are ready to write down the mQME for our model.

\begin{definition}
  \label{mqmetqmdfntn36rlylstlbl}Given $I_0$ in
  (\ref{scl0itactnblkbdrtrmmycnvntn}), we further require
  \begin{equation}
    J^{\partial} \in \mathcal{O} (L) [[\hbar]] . \label{plrztncnsttonbdrtrm}
  \end{equation}
  Then, we say that $I_t$ in (\ref{intactnatfnttscldfntnmm}) satisfies
  {\tmstrong{modified quantum master equation (mQME) at scale $t$}} if there
  exists
  \[ H^{\partial} \in \mathcal{O} (V) [[\hbar]], \quad | H^{\partial} | = 1,
     H^{\partial} \star_{\hbar} H^{\partial} = 0, \]
  such that
  \begin{equation}
    (1 \otimes (\hbar \mathd + \hbar^2 \partial_{K_t})
    +\tmmathbf{\Omega}_{L'}^{\tmop{left}} (H^{\partial}, -) \otimes 1)
    \mathbb{I}_{\theta_t} (e^{(\pi^{\ast} (\tmmathbf{\alpha}) + I_t) / \hbar})
    = 0, \label{mqmemqmemqme322}
  \end{equation}
  where $\star_{\hbar}$ and $\tmmathbf{\Omega}_{L'}^{\tmop{left}} (-, -)$ are
  defined in Section \ref{mylprdctwlqtztnalgprlmlry}, $\tmmathbf{\alpha}$ is
  defined in (\ref{spclbdrfctnlmycnvtn}).
  $\tmmathbf{\Omega}_{L'}^{\tmop{left}} (H^{\partial}, -)$ is called the
  {\tmstrong{BFV operator}}. In equation (\ref{mqmemqmemqme322}),
  $\mathbb{I}_{\theta_t}$ is well defined on $\mathcal{O} (\mathcal{E}_c)$
  because the image of $\theta_t$ is a subset of $\mathcal{E}_c$.
\end{definition}

If $I_{\varepsilon}$ satisfies mQME at scale $\varepsilon$, Theorem
\ref{ppstnrgrnofspltgmycvtn} together with (\ref{RGflowofrnmlzditactn321})
implies $I_{\Lambda}$ satisfies mQME at scale $\Lambda$, with the same BFV
operator. We then say that the data $(I^{\partial}, J^{\partial},
H^{\partial})$ defines a consistent interactive TQM on $\mathbb{R}_{\geqslant
0}$ in BV-BFV formalism.

\begin{remark}
  The condition (\ref{plrztncnsttonbdrtrm}) corresponds to
  {\cite[(2.29)]{cattaneo2018perturbative}}, which reflects the fact that the
  action could be modified by proper boundary terms according to the
  polarization we choose.
\end{remark}

\subsubsection*{Generic solutions to mQME}

To solve the mQME, we first study the BV anomaly of the interaction.

\begin{lemma}
  \label{bvanmlcmptnn37}Given $I_t$ defined in
  (\ref{intactnatfnttscldfntnmm}), we have
  \begin{eqnarray*}
    &  & (\hbar \mathd + \hbar^2 \partial_{K_t}) e^{I_t / \hbar}\\
    & = & \lim_{\varepsilon \rightarrow 0} e^{\hbar \partial_{P (\varepsilon,
    t)}} \left( e^{I_0 / \hbar} \left( \pi^{\ast} \left( e^{- J^{\partial} /
    \hbar} e^{\hbar \partial_{K_+^{\partial}}}_{\tmop{cross}} \left( e^{\hbar
    \partial_{(K^{\partial}_+ - K^{\partial}_-) / 4}} I^{\partial},
    e^{J^{\partial} / \hbar} \right) \right) +  \hbar^{- 1}
    \int_{\mathbb{R}_{\geqslant 0}} I^{\partial} \star_{\hbar} I^{\partial}
    \right) \right)
  \end{eqnarray*}
  where the limit in RHS is well defined by the same reason for Proposition
  \ref{tqmuvfntthmprf}. It is straightforward to see that both sides are
  products of $e^{I_t / \hbar}$ with an element in $\mathcal{O}
  (\mathcal{E}_c) [[\hbar]]$.
\end{lemma}

\begin{proof}
  Proposition \ref{extddppgtrltnwithlmtofppgtmycvtn} says that each $P
  (\varepsilon, t)$ in $e^{I_t / \hbar} \assign \lim_{\varepsilon \rightarrow
  0} e^{\hbar \partial_{P (\varepsilon, t)}} e^{I_0 / \hbar}$ can be replaced
  by a linear combination of $\bar{P} (0, t)$. So we use the following
  simplified notation:
  \begin{equation}
    e^{\hbar \partial_{\lfloor \bar{P} (0, t) \rfloor}} e^{I_0 / \hbar}
    \assign \lim_{\varepsilon \rightarrow 0} e^{\hbar \partial_{P
    (\varepsilon, t)}} e^{I_0 / \hbar} \label{rplcmtnttnfntsclactnmynttn}
  \end{equation}
  where we should regard the terms in LHS as integrations over spaces defined
  in (\ref{cnfgspcmyversnpts}).
  
  By (\ref{bvknrlisdextddppgtonconfgspc}), we have
  \[ (\hbar \mathd + \hbar^2 \partial_{K_t}) e^{I_t / \hbar} = (\hbar \mathd +
     \hbar^2 \partial_{\mathd \lfloor \bar{P} (0, t) \rfloor}) e^{\hbar
     \partial_{\lfloor \bar{P} (0, t) \rfloor}} e^{I_0 / \hbar}, \]
  which means the contraction with $K_t$ can be replaced by contraction with
  $\mathd$ of linear combination of $\bar{P} (0, t)$, while the replacement
  rule is the same as in (\ref{rplcmtnttnfntsclactnmynttn}).\footnote{For
  example, if $K_t$ contracts with two factors of a single bulk vertice, it
  will contribute as $K_t |_{\mathD_2}$ on $\mathbb{R}_{\geqslant 0}$ where
  $\mathD_2$ is the diagonal of $\mathbb{R}_{\geqslant 0}^2$. It is direct to
  verify $K_t |_{\mathD_2} = \tmop{dx} \partial_x (\bar{P} (0, t) (x, x)_{C_1}
  + \bar{P} (0, t) (x, x)_{C_2}) / 2$.} So, let $\varphi \in
  \mathcal{E}_c^{\otimes r}$ be some input field ($r > 0$),
  \begin{eqnarray*}
    &  & (\hbar \mathd + \hbar^2 \partial_{K_t}) e^{I_t / \hbar} (\varphi)\\
    & = & (\hbar \mathd + \hbar^2 \partial_{\mathd \lfloor \bar{P} (0, t)
    \rfloor}) e^{\hbar \partial_{\lfloor \bar{P} (0, t) \rfloor}}
    \sum_{\text{certain expansion}} \left( \text{$\tmop{certain}
    \tmop{coefficient}$} \right)\\
    &  & \hspace{6em} \times \int_{\mathbb{R}_{\geqslant 0} [n]}
    (I^{\partial})^n (\pi^{\ast} (J^{\partial}))^m (\varphi)\\
    & = & \sum_{\tmscript{\begin{array}{c}
      \tmop{certain} \tmop{expansion}\\
      \tmop{certain} \tmop{contraction}
    \end{array}}} \left( \text{$\tmop{certain} \tmop{coefficient}$} \right)\\
    &  & \hspace{6em} \times \int_{\mathbb{R}_{\geqslant 0} [n]}
    (I^{\partial})^n (\pi^{\ast} (J^{\partial}))^m ((\hbar \lfloor \bar{P} (0,
    t) \rfloor)^s \wedge \mathd \varphi + \mathd (\hbar \lfloor \bar{P} (0, t)
    \rfloor)^s \wedge \varphi)\\
    & = & e^{\hbar \partial_{\lfloor \bar{P} (0, t) \rfloor}}
    \sum_{\text{certain expansion}} \left( \text{$\tmop{certain}
    \tmop{coefficient}$} \right) \int_{\partial \mathbb{R}_{\geqslant 0} [n]}
    (I^{\partial})^n (\pi^{\ast} (J^{\partial}))^m (\varphi),
  \end{eqnarray*}
  where we have used Leibniz rule and Stokes' theorem for the last equality.
  More precisely,
  \begin{eqnarray*}
    &  & (\hbar \mathd + \hbar^2 \partial_{K_t}) e^{I_t / \hbar} (\varphi)\\
    & = & e^{I_t / \hbar} \sum_{\tmscript{\begin{array}{c}
      \text{contraction determined by}\\
      \text{connected graph } \Gamma
    \end{array}}} \left( \text{$\tmop{certain} \tmop{coefficient}$} \right)\\
    &  & \hspace{6em} \times \int_{\partial \mathbb{R}_{\geqslant 0} [|
    V_{\tmop{bulk}} (\Gamma) |]} (I^{\partial})^{| V_{\tmop{bulk}} (\Gamma) |}
    (\pi^{\ast} (J^{\partial}))^{| V_{\tmop{bdr}} (\Gamma) |} (\lfloor \bar{P}
    (0, t) \rfloor^{| E_{\tmop{int}} (\Gamma) |} \wedge \varphi)
  \end{eqnarray*}
  with $V_{\tmop{bulk}} (\Gamma), V_{\tmop{bdr}} (\Gamma), E_{\tmop{int}}
  (\Gamma)$ the set of bulk vertices, boundary vertices and internal edges of
  $\Gamma$, respectively.
  
  The boundary of $\mathbb{R}_{\geqslant 0} [| V_{\tmop{bulk}} (\Gamma) |]$
  has two kinds of components at codimension one. The first kind corresponds
  to a certain $x_i$ going to $0$, and the second kind corresponds to
  coincidence of some $x_i$ and $x_j$. So,
  \begin{eqnarray*}
    &  & (\hbar \mathd + \hbar^2 \partial_{K_t}) e^{I_t / \hbar}\\
    & = & \hbar e^{I_t / \hbar} (e^{\hbar \partial_{\lfloor \bar{P} (0, t)
    \rfloor}})_{\tmscript{\begin{array}{c}
      \text{contractions by}\\
      \text{connected graph}
    \end{array}}} e^{\pi^{\ast} (J^{\partial}) / \hbar}  \sum_{n = 1}^{+
    \infty} \hbar^{- n}\\
    &  & \left( I^{\partial} (x_1 = 0^+) \int_{x_2 \geqslant x_1}
    I^{\partial} (x_2) \ldots \int_{x_n \geqslant x_{n - 1}} I^{\partial}
    (x_n) \right.\\
    &  & + \sum_{i = 1}^{n - 1} \int_{x_1 \geqslant 0} I^{\partial} (x_1)
    \cdots \widehat{\int_{x_i}} \int_{x_{i + 1} \geqslant x_{i - 1}}
    (I^{\partial} (x_{i + 1}^-) - I^{\partial} (x_{i + 1}^+)) I^{\partial}
    (x_{i + 1}) \cdots \left. \int_{x_n \geqslant x_{n - 1}} I^{\partial}
    (x_n) \right)\\
    & = & e^{\hbar \partial_{\lfloor \bar{P} (0, t) \rfloor}} \left( e^{I_0 /
    \hbar} \left( \pi^{\ast} \left( e^{- J^{\partial} / \hbar} e^{\hbar
    \partial_{K_+^{\partial}}}_{\tmop{cross}} \left( e^{\hbar
    \partial_{(K^{\partial}_+ - K^{\partial}_-) / 4}} I^{\partial},
    e^{J^{\partial} / \hbar} \right) \right) +  \hbar^{- 1}
    \int_{\mathbb{R}_{\geqslant 0}} I^{\partial} \star_{\hbar} I^{\partial}
    \right) \right) .
  \end{eqnarray*}
  We have used (\ref{bdrdiffofextddppgt}) and (\ref{spclbdrvleofextddppgt}) to
  obtain the final result, while omitted details are left as an exercise.
\end{proof}

\begin{remark*}
  Similar BV anomaly computation for TQM on $S^1$ using configuration space
  can be found in {\cite[Section 3.4]{2017qtztnalgindexsiliqinli}}. For
  BF-like theories on manifold with boundary, {\cite[Section
  4.2]{cattaneo2018perturbative}} contains a rough description of such
  computations for partition functions. For Hamiltonian mechanics with constraints, {\cite[Appendix B]{Cattaneo_2022}} contains a proof of mQME for partition functions as well as a detailed consideration of the dependence of the solutions on gauge fixing conditions by using configuration space.
\end{remark*}

We can imagine an $\varepsilon \rightarrow 0$ version of Theorem
\ref{ppstnrgrnofspltgmycvtn}, which should imply:

\begin{lemma}
  \label{rltnlm38ok}Given $I_t$ defined in (\ref{intactnatfnttscldfntnmm}), we
  have
  \[ \lim_{\varepsilon \rightarrow 0} (1 \otimes e^{\hbar \partial_{P
     (\varepsilon, t)}}) e^{\left( \left( \int_{\mathbb{R}_{\geqslant 0}}
     I^{\partial} \right)_{\mathcal{E}_L} +\mathbb{I} (J^{\partial})
     +\mathbb{I} (\tmmathbf{\alpha}) \right) / \hbar} = e^{\mathbb{I}
     (\tmmathbf{\alpha}) / \hbar} \mathbb{I}_{\theta_t} (e^{I_t / \hbar}) . \]
\end{lemma}

The verification is left as an exercise.

Now we can write down generic solutions to mQME.

\begin{theorem}
  \label{thm310gnrcsltntqmrplus}Given
  \begin{eqnarray*}
    I^{\partial} & \in & \prod_{n \geqslant 2} \tmop{Sym}^n (V^{\ast}) \oplus
    \hbar \prod_{n > 0} \tmop{Sym}^n (V^{\ast}) [[\hbar]], \quad |
    I^{\partial} | = 1,\\
    J^{\partial} & \in & \prod_{n \geqslant 2} \tmop{Sym}^n (L^{\ast}) \oplus
    \hbar \prod_{n > 0} \tmop{Sym}^n (L^{\ast}) [[\hbar]], \quad |
    J^{\partial} | = 0,
  \end{eqnarray*}
  and $H^{\partial} \in \mathcal{O} (V) [[\hbar]]$ such that $| H^{\partial} |
  = 1, H^{\partial} \star_{\hbar} H^{\partial} = 0$, the mQME
  (\ref{mqmemqmemqme322}) is satisfied iff
  \[ I^{\partial} \star_{\hbar} I^{\partial} = 0 \quad \text{and} \quad
     H^{\partial} = - e^{J^{\partial} / \hbar} \star_{\hbar} I^{\partial}
     \star_{\hbar} e^{- J^{\partial} / \hbar} . \]
  (It is direct to check $e^{J^{\partial} / \hbar} \star_{\hbar} I^{\partial}
  \star_{\hbar} e^{- J^{\partial} / \hbar}$ is a well-defined element in
  $\mathcal{O} (V) [[\hbar]]$.)
\end{theorem}

\begin{proof}
  By (\ref{cptbltycdtnofspltgandbvdtmyycccvtn}) and Lemma
  \ref{bvanmlcmptnn37},
  \begin{eqnarray*}
    &  & (1 \otimes (\hbar \mathd + \hbar^2 \partial_{K_t}))
    \mathbb{I}_{\theta_t} (e^{(\pi^{\ast} (\tmmathbf{\alpha}) + I_t) /
    \hbar})\\
    & = & e^{\mathbb{I} (\tmmathbf{\alpha}) / \hbar} (1 \otimes (\hbar \mathd
    + \hbar^2 \partial_{K_t}) + \hbar \{ \mathbb{I} (\tmmathbf{\alpha}), -
    \}_t) \mathbb{I}_{\theta_t} (e^{I_t / \hbar})\\
    & = & e^{\mathbb{I} (\tmmathbf{\alpha}) / \hbar} \mathbb{I}_{\theta_t}
    ((\hbar \mathd + \hbar^2 \partial_{K_t}) e^{I_t / \hbar})\\
    & = & e^{\mathbb{I} (\tmmathbf{\alpha}) / \hbar} \mathbb{I}_{\theta_t}
    \Bigg(\nobracket \lim_{\varepsilon \rightarrow 0} e^{\hbar \partial_{P
    (\varepsilon, t)}} \Bigg(\nobracket e^{I_0 / \hbar}\\
    &  & \qquad \times \left. \left. \left( \pi^{\ast} \left( e^{-
    J^{\partial} / \hbar} e^{\hbar \partial_{K_+^{\partial}}}_{\tmop{cross}}
    \left( e^{\hbar \partial_{(K^{\partial}_+ - K^{\partial}_-) / 4}}
    I^{\partial}, e^{J^{\partial} / \hbar} \right) \right) +  \hbar^{- 1}
    \int_{\mathbb{R}_{\geqslant 0}} I^{\partial} \star_{\hbar} I^{\partial}
    \right) \right) \right) .
  \end{eqnarray*}
  By Lemma \ref{rltnlm38ok} and (\ref{usefulfactabta2mycvtn1}),
  \begin{eqnarray*}
    &  & (\tmmathbf{\Omega}_{L'}^{\tmop{left}} (H^{\partial}, -) \otimes 1)
    \mathbb{I}_{\theta_t} (e^{(\pi^{\ast} (\tmmathbf{\alpha}) + I_t) /
    \hbar})\\
    & = & (\tmmathbf{\Omega}_{L'}^{\tmop{left}} (H^{\partial}, -) \otimes 1)
    \lim_{\varepsilon \rightarrow 0} (1 \otimes e^{\hbar \partial_{P
    (\varepsilon, t)}}) e^{\left( \left( \int_{\mathbb{R}_{\geqslant 0}}
    I^{\partial} \right)_{\mathcal{E}_L} +\mathbb{I} (J^{\partial})
    +\mathbb{I} (\tmmathbf{\alpha}) \right) / \hbar}\\
    & = & \lim_{\varepsilon \rightarrow 0} (1 \otimes e^{\hbar \partial_{P
    (\varepsilon, t)}}) \left( e^{\left( \int_{\mathbb{R}_{\geqslant 0}}
    I^{\partial} \right)_{\mathcal{E}_L} / \hbar}
    (\tmmathbf{\Omega}_{L'}^{\tmop{left}} (H^{\partial}, -) \otimes 1)
    e^{(\mathbb{I} (J^{\partial}) +\mathbb{I} (\tmmathbf{\alpha})) / \hbar}
    \right)\\
    & = & \lim_{\varepsilon \rightarrow 0} (1 \otimes e^{\hbar \partial_{P
    (\varepsilon, t)}}) \left( e^{\left( \int_{\mathbb{R}_{\geqslant 0}}
    I^{\partial} \right)_{\mathcal{E}_L} / \hbar} e^{\mathbb{I}
    (\tmmathbf{\alpha}) / \hbar}  \right.\\
    &  & \hspace{8em} \times \left. \mathbb{I} \left( e^{- \hbar
    \partial_{K_-^{\partial}}}_{\tmop{cross}} \left( e^{\hbar
    \partial_{(K^{\partial}_+ - K^{\partial}_-) / 4}} H^{\partial},
    e^{J^{\partial} / \hbar} \right) \right) \right)\\
    & = & e^{\mathbb{I} (\tmmathbf{\alpha}) / \hbar} \mathbb{I}_{\theta_t}
    \left( \lim_{\varepsilon \rightarrow 0} e^{\hbar \partial_{P (\varepsilon,
    t)}} \left( e^{I_0 / \hbar} \times \pi^{\ast} \left( e^{- J^{\partial} /
    \hbar} e^{- \hbar \partial_{K_-^{\partial}}}_{\tmop{cross}} \left(
    e^{\hbar \partial_{(K^{\partial}_+ - K^{\partial}_-) / 4}} H^{\partial},
    e^{J^{\partial} / \hbar} \right) \right) \right) \right),
  \end{eqnarray*}
  with the last equality by the same reason as for Lemma \ref{rltnlm38ok}.
  
  So, the equation (\ref{mqmemqmemqme322}) is equivalent to
  \begin{eqnarray}
    I^{\partial} \star_{\hbar} I^{\partial} & = & 0, \nonumber\\
    - e^{\hbar \partial_{K_+^{\partial}}}_{\tmop{cross}} \left( e^{\hbar
    \partial_{(K^{\partial}_+ - K^{\partial}_-) / 4}} I^{\partial},
    e^{J^{\partial} / \hbar} \right) & = & e^{- \hbar
    \partial_{K_-^{\partial}}}_{\tmop{cross}} \left( e^{\hbar
    \partial_{(K^{\partial}_+ - K^{\partial}_-) / 4}} H^{\partial},
    e^{J^{\partial} / \hbar} \right) . 
    \label{mqmeeqn2bdrbfvoprtandblkactnrltn325}
  \end{eqnarray}
  Since we have imposed that $J^{\partial} \in \mathcal{O} (L) [[\hbar]]$, by
  the fact that $K^{\partial}_- \in L \otimes L'$,
  \[ e^{- \hbar \partial_{K_-^{\partial}}}_{\tmop{cross}} \left( e^{\hbar
     \partial_{(K^{\partial}_+ - K^{\partial}_-) / 4}} H^{\partial},
     e^{J^{\partial} / \hbar} \right) = e^{J^{\partial} / \hbar} e^{\hbar
     \partial_{(K^{\partial}_+ - K^{\partial}_-) / 4}} H^{\partial} . \]
  So (\ref{mqmeeqn2bdrbfvoprtandblkactnrltn325}) leads to
  \begin{eqnarray*}
    H^{\partial} & = & - e^{- \hbar \partial_{(K^{\partial}_+ -
    K^{\partial}_-) / 4}} \left( e^{- J^{\partial} / \hbar} e^{\hbar
    \partial_{K_+^{\partial}}}_{\tmop{cross}} \left( e^{\hbar
    \partial_{(K^{\partial}_+ - K^{\partial}_-) / 4}} I^{\partial},
    e^{J^{\partial} / \hbar} \right) \right)\\
    & = & - e^{\hbar \partial_{K^{\partial} / 2}}_{\tmop{cross}} \left( e^{-
    J^{\partial} / \hbar},  \left( e^{\hbar \partial_{K^{\partial} /
    2}}_{\tmop{cross}} (I^{\partial}, e^{J^{\partial} / \hbar}) \right)
    \right)\\
    & = & - e^{J^{\partial} / \hbar} \star_{\hbar} I^{\partial} \star_{\hbar}
    e^{- J^{\partial} / \hbar} .
  \end{eqnarray*}
  
\end{proof}

\begin{remark}
  The condition $I^{\partial} \star_{\hbar} I^{\partial} = 0$ is equivalent to
  $\int_{S^1} I^{\partial}$ determining a consistent interaction for TQM on
  $S^1$ in perturbative BV formalism (see {\cite[Theorem
  3.10]{2017qtztnalgindexsiliqinli}}).
\end{remark}

\begin{remark}
  Actually the $J^{\partial} \neq 0$ solution can be obtained by a ``boundary
  gauge transformation'' from the $J^{\partial} = 0$ solution: if
  $I^{\partial} \star_{\hbar} I^{\partial} = 0$,
  \begin{eqnarray*}
    0 & = & (1 \otimes (\hbar \mathd + \hbar^2 \partial_{K_t})
    -\tmmathbf{\Omega}_{L'}^{\tmop{left}} (I^{\partial}, -) \otimes 1)
    \lim_{\varepsilon \rightarrow 0} (1 \otimes e^{\hbar \partial_{P
    (\varepsilon, t)}}) e^{\left( \left( \int_{\mathbb{R}_{\geqslant 0}}
    I^{\partial} \right)_{\mathcal{E}_L} +\mathbb{I} (\tmmathbf{\alpha})
    \right) / \hbar}\\
    &  & \\
    & = & (1 \otimes (\hbar \mathd + \hbar^2 \partial_{K_t})
    -\tmmathbf{\Omega}_{L'}^{\tmop{left}} (e^{J^{\partial} / \hbar}
    \star_{\hbar} I^{\partial} \star_{\hbar} e^{- J^{\partial} / \hbar}, -)
    \otimes 1)\\
    &  & \hspace{7em} \lim_{\varepsilon \rightarrow 0} (1 \otimes e^{\hbar
    \partial_{P (\varepsilon, t)}}) (\tmmathbf{\Omega}_{L'}^{\tmop{left}}
    (e^{J^{\partial} / \hbar}, -) \otimes 1) e^{\left( \left(
    \int_{\mathbb{R}_{\geqslant 0}} I^{\partial} \right)_{\mathcal{E}_L}
    +\mathbb{I} (\tmmathbf{\alpha}) \right) / \hbar}\\
    & = & (1 \otimes (\hbar \mathd + \hbar^2 \partial_{K_t})
    -\tmmathbf{\Omega}_{L'}^{\tmop{left}} (e^{J^{\partial} / \hbar}
    \star_{\hbar} I^{\partial} \star_{\hbar} e^{- J^{\partial} / \hbar}, -)
    \otimes 1)\\
    &  & \hspace{7em} \lim_{\varepsilon \rightarrow 0} (1 \otimes e^{\hbar
    \partial_{P (\varepsilon, t)}}) e^{\left( \left(
    \int_{\mathbb{R}_{\geqslant 0}} I^{\partial} \right)_{\mathcal{E}_L}
    +\mathbb{I} (J^{\partial}) +\mathbb{I} (\tmmathbf{\alpha}) \right) /
    \hbar} .
  \end{eqnarray*}
\end{remark}

\section{BV-BFV Description of TQM on
Interval}\label{bvbfvdscptntqmintvlsctn4ttl}

Now, consider TQM on interval $\mathbf{I}= [0, 1]$. The spacetime is compact,
so we are free of those subtleties in Remark
\ref{rmkncmpctpprsrptfctsitactvthy}. We will be sketchy since the story is
similar to that in Section \ref{sctnbvbfvtqmpstvrllne3}.

\subsubsection*{Field spaces}

We have a free BV-BFV pair with bulk field data
\begin{equation}
  \left( \mathcal{E} \assign \Omega^{\bullet} (\mathbf{I}) \otimes V, \mathd,
  \omega \assign \int_{\mathbf{I}} \omega^{\partial} \right)
  \label{blkflddtintvl41}
\end{equation}
and boundary field data
\begin{equation}
  (\mathcal{E}^{\partial} \assign V \oplus V = (\Omega^{\bullet} \{ 0 \}
  \otimes V) \oplus (\Omega^{\bullet} \{ 1 \} \otimes V), 0, \omega^{\partial}
  \oplus (- \omega^{\partial})) \label{bdrflddt42}
\end{equation}
($\Omega^{\bullet} \{ 0 \}$ and $\Omega^{\bullet} \{ 1 \}$ denote the de Rham
complexes of the boundary points), and the restriction map
\begin{equation}
  \pi \assign (\pi_0, \pi_1), \label{bdrrstrctnmp43333}
\end{equation}
where $\pi_0, \pi_1$ are pullbacks of forms induced by $\{ 0 \}, \{ 1 \}
\hookrightarrow \mathbf{I}$, respectively. Now $(V, \omega^{\partial})$ is a
finite dimensional graded symplectic vector space endowed with two Lagrangian
decompositions
\begin{equation}
  V = L_0 \oplus L'_0 = L_1 \oplus L'_1, \label{twolgrgndcmpstn43}
\end{equation}
by which we can decompose $K^{\partial}$ (the inverse to $\omega^{\partial}$,
see (\ref{invrssplctc27})) as
\[ K^{\partial} = K^{\partial}_{0, -} + K^{\partial}_{0, +} = K^{\partial}_{1,
   -} + K^{\partial}_{1, +}, \quad K^{\partial}_{i, -} \in L_i \otimes L_i',
   \enspace K^{\partial}_{i, +} \in L_i' \otimes L_i \text{ for } i = 0, 1. \]
The restricted bulk field space with boundary condition $L_0$ at $\{ 0 \}$ and
$L_1$ at $\{ 1 \}$ is:
\begin{equation}
  \mathcal{E}_{L_0, L_1} \assign \left\{ f \in \Omega^{\bullet} (\mathbf{I})
  \otimes V| \pi_i (f) \in L_i \text{ for } i = 0, 1 \right\} .
  \label{whatihvnlblthisdfntn44}
\end{equation}
Then, $(\mathcal{E}_{L_0, L_1}, \mathd, \omega)$ is a dg $(- 1)$-symplectic
vector space.

\subsubsection*{Renormalized BV structure with restricted bulk field space}

Now we formulate the renormalized BV structure on $\mathcal{O}
(\mathcal{E}_{L_0, L_1})$. For Lagrangian decomposition $V = L_0 \oplus L'_0$,
we denote the propagator defined in (\ref{ppgtdfntnbvbfvwrkmycvtn}) as:
\begin{eqnarray*}
  P_{L_0}^{\mathbb{R}_{\geqslant 0}} (\varepsilon, \Lambda) & \assign & \left(
  \frac{- 1}{4}  \left( \text{d}^{\tmop{GF}} \otimes 1 + 1 \otimes
  \mathd^{\tmop{GF}} \right) \int_{\varepsilon}^{\Lambda} \tmop{dt}
  \widetilde{H_t} \right) \otimes K^{\partial}_{0, +}\\
  &  & + \sigma \left( \frac{1}{4}  \left( \text{d}^{\tmop{GF}} \otimes 1 + 1
  \otimes \mathd^{\tmop{GF}} \right) \int_{\varepsilon}^{\Lambda} \tmop{dt}
  \widetilde{H_t} \right) \otimes K^{\partial}_{0, -} 
\end{eqnarray*}
with
\[ \widetilde{H_t} = \frac{1}{\sqrt{4 \pi t}} \left( \phi (x - y) e^{-
   \frac{(x - y)^2}{4 t}} (\tmop{dy} - \tmop{dx}) - \phi (x + y) e^{- \frac{(x
   + y)^2}{4 t}} (\tmop{dy} + \tmop{dx}) \right) . \]
Suppose $\phi$ is supported on a subset of the open region $(- 0.1, 0.1)$.
Then it is direct to see that, on $(\mathbb{R}_{\geqslant 0} \times
\mathbb{R}_{\geqslant 0}) \backslash ([0, 0.1] \times [0, 0.1])$,
\[ P_{L_0}^{\mathbb{R}_{\geqslant 0}} (\varepsilon, \Lambda) = \left( \frac{-
   1}{4}  \left( \text{d}^{\tmop{GF}} \otimes 1 + 1 \otimes \mathd^{\tmop{GF}}
   \right) \int_{\varepsilon}^{\Lambda} \tmop{dt} \frac{1}{\sqrt{4 \pi t}}
   \phi (x - y) e^{- \frac{(x - y)^2}{4 t}} (\tmop{dy} - \tmop{dx}) \right)
   \otimes K^{\partial} . \]
Similarly we can write down a propagator for Lagrangian decomposition $V = L_1
\oplus L'_1$ with the spacetime modified to $\mathbb{R}_{\leqslant 1} \assign
(- \infty, 1]$:
\begin{eqnarray*}
  P_{L_1}^{\mathbb{R}_{\leqslant 1}} (\varepsilon, \Lambda) & \assign & \left(
  \frac{- 1}{4}  \left( \text{d}^{\tmop{GF}} \otimes 1 + 1 \otimes
  \mathd^{\tmop{GF}} \right) \int_{\varepsilon}^{\Lambda} \tmop{dt}
  \widetilde{H'_t} \right) \otimes K^{\partial}_{1, +}\\
  &  & + \sigma \left( \frac{1}{4}  \left( \text{d}^{\tmop{GF}} \otimes 1 + 1
  \otimes \mathd^{\tmop{GF}} \right) \int_{\varepsilon}^{\Lambda} \tmop{dt}
  \widetilde{H'_t} \right) \otimes K^{\partial}_{1, -} 
\end{eqnarray*}
with
\[ \widetilde{H'_t} = \frac{1}{\sqrt{4 \pi t}} \left( \phi (x - y) e^{-
   \frac{(x - y)^2}{4 t}} (\tmop{dy} - \tmop{dx}) - \phi (x + y - 2) e^{-
   \frac{(x + y - 2)^2}{4 t}} (\tmop{dy} + \tmop{dx}) \right) . \]
Then, on $(\mathbb{R}_{\leqslant 1} \times \mathbb{R}_{\leqslant 1})
\backslash ([0.9, 1] \times [0.9, 1])$,
\[ P_{L_1}^{\mathbb{R}_{\leqslant 1}} (\varepsilon, \Lambda) = \left( \frac{-
   1}{4}  \left( \text{d}^{\tmop{GF}} \otimes 1 + 1 \otimes \mathd^{\tmop{GF}}
   \right) \int_{\varepsilon}^{\Lambda} \tmop{dt} \frac{1}{\sqrt{4 \pi t}}
   \phi (x - y) e^{- \frac{(x - y)^2}{4 t}} (\tmop{dy} - \tmop{dx}) \right)
   \otimes K^{\partial} . \]
So, we can define $P_{L_0, L_1} (\varepsilon, \Lambda)$ on $[0, 1] \times [0,
1]$ by gluing:
\begin{equation}
  P_{L_0, L_1} (\varepsilon, \Lambda) \assign \left\{\begin{array}{l}
    P_{L_1}^{\mathbb{R}_{\leqslant 1}} (\varepsilon, \Lambda) \quad \text{on }
    (\mathbf{I} \times \mathbf{I}) \backslash ([0, 0.1] \times [0, 0.1])\\
    P_{L_0}^{\mathbb{R}_{\geqslant 0}} (\varepsilon, \Lambda) \quad \text{on }
    (\mathbf{I} \times \mathbf{I}) \backslash ([0.9, 1] \times [0.9, 1])
  \end{array}\right. . \label{ppgtintvltqm44}
\end{equation}
This is the propagator from scale $\varepsilon$ to scale $\Lambda$ of the
renormalized free theory on $\mathcal{O} (\mathcal{E}_{L_0, L_1})$. Similarly
we can define the BV kernel at scale $t$ by gluing, denoted by $K_{L_0, L_1,
t}$ (explicit formula omitted). Just as in Proposition
\ref{extddppgtrltnwithlmtofppgtmycvtn}, we can also define the extended
propagator
\begin{equation}
  \overline{P_{L_0, L_1}} (0, t) \in \Omega^0 (\mathbf{I} [2]) \otimes
  V^{\otimes 2} \label{extddppgttqmintvl44}
\end{equation}
with $\mathbf{I} [2] \assign \{ (x, y) |0 \leqslant x \leqslant y \leqslant 1
\} \sqcup \{ (x, y) |0 \leqslant y \leqslant x \leqslant 1 \}$. Then $K_{L_0,
L_1, t} = \mathd \overline{P_{L_0, L_1}} (0, t)$ on $\mathbf{I} [2]$.

In summary, $\left( \mathcal{O} (\mathcal{E}_{L_0, L_1}), \mathd,
\partial_{K_{L_0, L_1, t}} \right)$ is a differential BV algebra for $t > 0$,
and we have the following conjugation:
\begin{equation}
  \left( \mathcal{O} (\mathcal{E}_{L_0, L_1}) [[\hbar]], \mathd + \hbar
  \partial_{K_{L_0, L_1, \varepsilon}} \right) \begin{array}{c}
    e^{\hbar \partial_{P_{L_0, L_1} (\varepsilon, \Lambda)}}\\
    \rightleftharpoons\\
    e^{- \hbar \partial_{P_{L_0, L_1} (\varepsilon, \Lambda)}}
  \end{array} \left( \mathcal{O} (\mathcal{E}_{L_0, L_1}) [[\hbar]], \mathd +
  \hbar \partial_{K_{L_0, L_1, \Lambda}} \right) .
  \label{freerstctdthyrghmtpyrltn44}
\end{equation}
If we replace $\mathcal{O} (\mathcal{E}_{L_0, L_1})$ in
(\ref{freerstctdthyrghmtpyrltn44}) with $\mathcal{O} (\mathcal{E})$ where
$\mathcal{E}$ is defined in (\ref{blkflddtintvl41}), the resulting relation
still holds.

\subsubsection*{Renormalized splitting}

Similar to (\ref{nttnspcliso316}), for $\forall J \in \mathcal{O} (V), i = 0,
1$, define $\mathbb{I}_i (J) \in \mathcal{O} (L_i') \otimes \mathcal{O}
(\mathcal{E}_{L_0, L_1})$ to be
\begin{equation}
  \mathbb{I}_i (J) \assign \left( 1 \otimes \left( \pi_i |_{\mathcal{E}_{L_0,
  L_1}} \right)^{\ast} \right) (J), \label{lstfmlisweargdfk491003}
\end{equation}
where we regard $J$ as in $\mathcal{O} (L'_i) \otimes \mathcal{O} (L_i)$ and
$\pi_i |_{\mathcal{E}_{L_0, L_1}} : \mathcal{E}_{L_0, L_1} \rightarrow L_i$ is
the restriction of $\pi_i$ in (\ref{bdrrstrctnmp43333}) to $\mathcal{E}_{L_0,
L_1} \subset \mathcal{E}$. Now, for any splitting $\theta : L_0' \oplus L'_1
\rightarrow \mathcal{E}$ as in Definition \ref{plrztnandspltgmycvtn} and the
induced $\mathbb{I}_{\theta} : \mathcal{O} (\mathcal{E}) \rightarrow
\mathcal{O} (L_0') \otimes \mathcal{O} (L_1') \otimes \mathcal{O}
(\mathcal{E}_{L_0, L_1})$ in (\ref{ismbtwnfctnrgonfldspc}), it is direct to
see that $\mathbb{I}_{\theta} (\pi^{\ast}_i (J)) =\mathbb{I}_i (J)$.

For $i = 0, 1$, let $\tmmathbf{\alpha}_i \in (L'_i)^{\ast} \otimes L_i^{\ast}$
be determined by
\begin{equation}
  \tmmathbf{\alpha}_i (l' \otimes l) \assign \omega^{\partial} (l', l) \text{
  for } l' \in L'_i, l \in L_i . \label{spclbdrtrmtqmintvl2tms47}
\end{equation}
Regard $\tmmathbf{\alpha}_i$ as an element in $\tmop{Sym}^2 (V^{\ast})$. Then,
\[ \mathbb{I}_i (\tmmathbf{\alpha}_i) \in (L_i')^{\ast} \otimes
   (\mathcal{E}_{L_0, L_1})^{\ast}, \quad \mathbb{I}_i (\tmmathbf{\alpha}_i)
   (l' \otimes f) = \omega^{\partial} (l', \pi_i (f)) . \]
Similar to Definition \ref{dfntn34rnzdsplttglstlt}, we define:

\begin{definition}
  Given the extended propagator in (\ref{extddppgttqmintvl44}) and
  $\tmmathbf{\alpha}_0, \tmmathbf{\alpha}_1$ in
  (\ref{spclbdrtrmtqmintvl2tms47}), the scale $t$ renormalized splitting
  $\theta_t : L'_0 \oplus L'_1 \rightarrow \mathcal{E}$ for the current TQM on
  $\mathbf{I}$ is
  \begin{eqnarray}
    &  & \theta_t (l'_0, l'_1) \nonumber\\
    & \assign & 2 (\mathbb{I}_0 (\tmmathbf{\alpha}_0) (l'_0 \otimes -)
    \otimes 1) \overline{P_{L_0, L_1}} (0, t) |_{C_1} - 2 (\mathbb{I}_1
    (\tmmathbf{\alpha}_1) (l'_1 \otimes -) \otimes 1) \overline{P_{L_0, L_1}}
    (0, t) |_{C_2}  \label{scltspltgitvcfgrtn47}
  \end{eqnarray}
  where $l'_0 \in L'_0, l'_1 \in L'_1$, $C_1 \assign \{ (x, y) |0 \leqslant x
  \leqslant y \leqslant 1 \}, C_2 \assign \{ (x, y) | 0 \leqslant y \leqslant
  x \leqslant 1 \}$ here.
\end{definition}

The renormalization of $\theta_t$ reads:

\begin{theorem}
  \label{ppstnrgrnofspltgintvlppstn41}For $\forall \varepsilon, \Lambda > 0$,
  the renormalized splittings $\theta_{\varepsilon}, \theta_{\Lambda}$ in
  (\ref{scltspltgitvcfgrtn47}) and the conjugation map $e^{\hbar
  \partial_{P_{L_0, L_1} (\varepsilon, \Lambda)}}$ in
  (\ref{freerstctdthyrghmtpyrltn44}) make this diagram commute:
  \begin{eqnarray*}
    \mathcal{O} (\mathcal{E}) [[\hbar]] \hspace{2.5em} &
    \xrightarrow{\LARGE{\mathbb{I}_{\theta_{\varepsilon}}}} & \quad
    \mathcal{O} (L_0') \otimes \mathcal{O} (L_1') \otimes \mathcal{O}
    (\mathcal{E}_{L_0, L_1}) [[\hbar]]\\
    \downarrow \enspace e^{\hbar \partial_{P_{L_0, L_1} (\varepsilon,
    \Lambda)}} &  & \hspace{1.5em} \downarrow \enspace e^{- (\mathbb{I}_0
    (\tmmathbf{\alpha}_0) -\mathbb{I}_1 (\tmmathbf{\alpha}_1)) / \hbar} \left(
    1 \otimes 1 \otimes e^{\hbar \partial_{P_{L_0, L_1} (\varepsilon,
    \Lambda)}} \right) e^{(\mathbb{I}_0 (\tmmathbf{\alpha}_0) -\mathbb{I}_1
    (\tmmathbf{\alpha}_1)) / \hbar}\\
    \mathcal{O} (\mathcal{E}) [[\hbar]] \hspace{2.5em} &
    \xrightarrow{\LARGE{\mathbb{I}_{\theta_{\Lambda}}}} & \quad \mathcal{O}
    (L_0') \otimes \mathcal{O} (L_1') \otimes \mathcal{O} (\mathcal{E}_{L_0,
    L_1}) [[\hbar]]
  \end{eqnarray*}
\end{theorem}

This completes the description of homotopic renormalization for free TQM on
$\mathbf{I}$ in BV-BFV formalism.

\subsubsection*{The mQME and its generic solutions}

Using the extended propagator (\ref{extddppgttqmintvl44}), we can repeat the
proof of Proposition \ref{tqmuvfntthmprf} to show that, given
\begin{eqnarray}
  I^{\partial} & \in & \prod_{n \geqslant 2} \tmop{Sym}^n (V^{\ast}) \oplus
  \hbar \prod_{n > 0} \tmop{Sym}^n (V^{\ast}) [[\hbar]], \quad | I^{\partial}
  | = 1, \nonumber\\
  J^{\partial}_0 & \in & \prod_{n \geqslant 2} \tmop{Sym}^n (L^{\ast}_0)
  \oplus \hbar \prod_{n > 0} \tmop{Sym}^n (L^{\ast}_0) [[\hbar]], \quad |
  J^{\partial}_0 | = 0, \nonumber\\
  J^{\partial}_1 & \in & \prod_{n \geqslant 2} \tmop{Sym}^n (L^{\ast}_1)
  \oplus \hbar \prod_{n > 0} \tmop{Sym}^n (L^{\ast}_1) [[\hbar]], \quad |
  J^{\partial}_1 | = 0,  \label{fmldfntn410pleeesend}
\end{eqnarray}
the scale $t$ interaction
\begin{equation}
  I_t \assign \lim_{\varepsilon \rightarrow 0} \hbar \log (e^{\hbar
  \partial_{P (\varepsilon, t)}} e^{I_0 / \hbar}) \label{scltinactntqmintvl49}
\end{equation}
with $I_0 \assign \pi^{\ast}_0 (J^{\partial}_0) + \pi^{\ast}_1
(J^{\partial}_1) + \int_{\mathbf{I}} I^{\partial}$ is a well-defined element
in $\mathcal{O} (\mathcal{E}) [[\hbar]]$. Let
\begin{equation}
  H^{\partial}_0, H^{\partial}_1 \in \mathcal{O} (V) [[\hbar]], \quad |
  H^{\partial}_i | = 1, H^{\partial}_i \star_{\hbar} H^{\partial}_i = 0 \text{
  for } i = 0, 1. \label{fml412gdgdgdeeend}
\end{equation}
They induces a BFV operator
\begin{equation}
  \tmmathbf{\Omega}_{L_0'}^{\tmop{left}} (H^{\partial}_0, -) \otimes 1 + 1
  \otimes \tmmathbf{\Omega}_{L_1'}^{\tmop{right}} (-, H_1^{\partial})
  \label{bfvoprtrintvlcs415plsend}
\end{equation}
on $\mathcal{O} (L_0') \otimes \mathcal{O} (L_1') [[\hbar]]$, where
$\tmmathbf{\Omega}_{L_1'}^{\tmop{right}} (-, -)$ is defined in
(\ref{weylqtztnrightactnmycnvntn}) with the substitution $L \rightarrow L_1',
L' \rightarrow L_1$. Similar to Definition \ref{mqmetqmdfntn36rlylstlbl}, we
define:

\begin{definition}
  \label{mqmeintvlcsdfntn42plslstlbl}Given $I^{\partial}, J^{\partial}_0,
  J^{\partial}_1$ as in (\ref{fmldfntn410pleeesend}), we say that $I_t$ in
  (\ref{scltinactntqmintvl49}) satisfies scale $t$ modified quantum master
  equation (mQME) with (\ref{bfvoprtrintvlcs415plsend}) being the BFV operator
  if
  \begin{eqnarray}
    &  & ((\tmmathbf{\Omega}_{L_0'}^{\tmop{left}} (H^{\partial}_0, -) \otimes
    1 + 1 \otimes \tmmathbf{\Omega}_{L_1'}^{\tmop{right}} (-, H_1^{\partial}))
    \otimes 1 \nobracket \nonumber\\
    &  & \hspace{8em} \left. + 1 \otimes 1 \otimes \left( \hbar \mathd +
    \hbar^2 \partial_{K_{L_0, L_1, t}} \right) \right) \mathbb{I}_{\theta_t}
    (e^{(\pi^{\ast}_0 (\tmmathbf{\alpha}_0) - \pi^{\ast}_1
    (\tmmathbf{\alpha}_1) + I_t) / \hbar}) = 0.  \label{mqmeintvltqmmycvtn49}
  \end{eqnarray}
\end{definition}

By arguments similar to those in Section
\ref{tctntofitctnthytqmpstvrllnsctn32}, we have:

\begin{theorem}
  \label{ppstn42gnrcsltntqmintvllbbububu}Given $I^{\partial}, J^{\partial}_0,
  J^{\partial}_1$ as in (\ref{fmldfntn410pleeesend}), $H^{\partial}_0,
  H^{\partial}_1$ as in (\ref{fml412gdgdgdeeend}), the mQME
  (\ref{mqmeintvltqmmycvtn49}) is satisfied iff
  \begin{equation}
    I^{\partial} \star_{\hbar} I^{\partial} = 0, \quad H_0^{\partial} = -
    e^{J^{\partial}_0 / \hbar} \star_{\hbar} I^{\partial} \star_{\hbar} e^{-
    J_0^{\partial} / \hbar}, \quad H_1^{\partial} = e^{- J^{\partial}_1 /
    \hbar} \star_{\hbar} I^{\partial} \star_{\hbar} e^{J_1^{\partial} / \hbar}
    . \label{gnrcsltntqmitvl411}
  \end{equation}
\end{theorem}

This describes generic solutions to the mQME for TQM on interval.

\begin{example}
  \label{eg37thycnstctn}``BF theory with B-A boundary condition''
  
  Let $\mathfrak{g}$ be a Lie algebra with basis $\{ t^a \}_{a = 1}^{\ell}$,
  $[t^a, t^b] = f^{a b}_c t^c$.
  
  For $\beta \in \mathfrak{g}^{\ast}$, we use $\epsilon \beta$ to denote the
  element in $(\mathfrak{g}^{\ast}) [- 1]$ corresponding to $\beta$, where
  $\epsilon$ is a formal variable of degree $1$; similarly for $\alpha \in
  \mathfrak{g}$ we have $\eta \alpha \in \mathfrak{g} [1]$ where $\eta$ is a
  formal variable of degree $- 1$. There is a degree $0$ symplectic pairing
  $\omega^{\partial}$ on
  \[ V \assign (\mathfrak{g}^{\ast}) [- 1] \oplus \mathfrak{g} [1], \]
  determined by
  \[ \omega^{\partial} (\epsilon t_a, \eta t^b) \assign \delta_a^b \]
  where $\{ t_a \}_{a = 1}^{\ell}$ is the basis of $\mathfrak{g}^{\ast}$ dual
  to $\{ t^a \}_{a = 1}^{\ell}$. Then, for the graded symplectic vector space
  $(V, \omega^{\partial})$, the inverse to $\omega^{\partial}$ (see
  (\ref{invrssplctc27}) for definition) is
  \[ K^{\partial} = - \eta t^a \otimes \epsilon t_a - \epsilon t_a \otimes
     \eta t^a . \]
  The model we talk about is the so-called ``BF theory'' on $\mathbf{I}$, with
  bulk field space $\left( \mathcal{E}= \Omega^{\bullet} (\mathbf{I}) \otimes
  V, \mathd, \omega = \int_{\mathbf{I}} \omega^{\partial} \right)$ and
  boundary field space $(\mathcal{E}^{\partial} = V \oplus V, 0,
  \omega^{\partial} \oplus (- \omega^{\partial}))$. We choose the following
  polarization:
  \begin{equation}
    L_0 = (\mathfrak{g}^{\ast}) [- 1], L'_0 =\mathfrak{g} [1] ; \quad L_1
    =\mathfrak{g} [1], L'_1 = (\mathfrak{g}^{\ast}) [- 1],
    \label{bfthyintvlabplrztn}
  \end{equation}
  which is the so-called ``B boundary condition'' at $\{ 0 \}$ and ``A
  boundary condition'' at $\{ 1 \}$.
  
  Let $\{ B^a \}_{a = 1}^{\ell}$ be the basis of $((\mathfrak{g}^{\ast}) [-
  1])^{\ast}$ and $\{ A_a \}_{a = 1}^{\ell}$ be the basis of $(\mathfrak{g}
  [1])^{\ast}$ so that
  \[ B^a (\epsilon t_b) = A_b (\eta t^a) = \delta_b^a . \]
  Then, BF theory corresponds to the following choice:
  \begin{equation}
    I^{\partial} = \frac{1}{2} f^{a b}_c B^c A_a A_b, \quad J^{\partial}_0 =
    J^{\partial}_1 = 0, \quad H^{\partial}_0 = - H^{\partial}_1 = -
    I^{\partial} . \label{bfthydtactnbdractnbfvoprt414}
  \end{equation}
  This data satisfies (\ref{gnrcsltntqmitvl411}) because
  \[ I^{\partial} \star_{\hbar} I^{\partial} = - \hbar f^{a b}_{b'} f^{a'
     b'}_{c'} B^{c'} A_a A_b A_{a'} = 0 \]
  by Jacobi identity. Hence we indeed obtain a consistent interactive TQM in
  BV-BFV formalism.
\end{example}

\section{Return to the BV Description}\label{sctn5555555lst}

In {\cite{wangyantqm202203}} we only work on the restricted bulk field space
$\mathcal{E}_L$, and define the theory within BV formalism. This approach is
taken in several other works to deal with QFT's on manifold with boundary,
such as
{\cite{rabinovich2021factorization,gwilliam2021factorization,gwilliam2019onelpqtzncs,zeng2021monopole}}.
Moreover, factorization algebras in the sense of
{\cite{costello_gwilliam_2016,costello_gwilliam_2021}} can be constructed on
the spacetime to describe observables of the theory (see e.g.,
{\cite{rabinovich2021factorization}}). In this section, we will extract such a
BV description from the BV-BFV description for TQM in previous sections. Then,
the mQME is translated to a condition arising from factorization algebra data.
This translation exhibits the connection between the BV-BFV interpretation and
the factorization algebra interpretation of the ``state''.

\subsubsection*{QME with restricted bulk field space}

We first consider TQM on $\mathbb{R}_{\geqslant 0}$ discussed in Section
\ref{sctnbvbfvtqmpstvrllne3}.

There is a differential BV algebra $(\mathcal{O} (\mathcal{E}_L), \mathd,
\partial_{K_t})$, where the restricted bulk field space $\mathcal{E}_L$ is
defined in (\ref{rstctdfldspctqmrpstvsctn3fml4}) and the BV kernel $K_t$ is
defined in (\ref{bvknldfntnbvbfvwrk}). The complex $(\mathcal{O}
(\mathcal{E}_L) [[\hbar]], \mathd + \hbar \partial_{K_t})$ will be perturbed
by restriction of interaction $I_t$ in (\ref{intactnatfnttscldfntnmm}) to
$\mathcal{E}_L$, which is
\[ I_t |_{\mathcal{E}_L} = \lim_{\varepsilon \rightarrow 0} \hbar \log \left(
   e^{\hbar \partial_{P (\varepsilon, t)}} e^{(I_0 |_{\mathcal{E}_L}) / \hbar}
   \right) \]
with
\begin{equation}
  I_0 |_{\mathcal{E}_L} =\mathbb{I} (J^{\partial}) + \left(
  \int_{\mathbb{R}_{\geqslant 0}} I^{\partial} \right)_{\mathcal{E}_L}
  \label{scl0itactntqmrrpstv51dftn}
\end{equation}
and $J^{\partial} \in \mathcal{O} (L) [[\hbar]]$ so that $\mathbb{I}
(J^{\partial}) \in \mathcal{O} (\mathcal{E}_L) [[\hbar]]$ ($\mathbb{I}$ is
defined in (\ref{nttnspcliso316})). Then, BV quantization amounts to consider
the following QME:
\begin{equation}
  (\mathd + \hbar \partial_{K_t}) e^{(I_t |_{\mathcal{E}_L}) / \hbar} = 0.
  \label{bvqmerstctdfldspctqmrplus51}
\end{equation}
By Lemma \ref{bvanmlcmptnn37}, it is direct to see
\begin{eqnarray*}
  &  & (\hbar \mathd + \hbar^2 \partial_{K_t}) e^{(I_t |_{\mathcal{E}_L}) /
  \hbar}\\
  & = & \lim_{\varepsilon \rightarrow 0} e^{\hbar \partial_{P (\varepsilon,
  t)}} \left( e^{(I_0 |_{\mathcal{E}_L}) / \hbar} \left( \mathbb{I}p_L \left(
  e^{- J^{\partial} / \hbar} e^{\hbar
  \partial_{K_+^{\partial}}}_{\tmop{cross}} \left( e^{\hbar
  \partial_{(K^{\partial}_+ - K^{\partial}_-) / 4}} I^{\partial},
  e^{J^{\partial} / \hbar} \right) \right) \right. \right.\\
  &  & \hspace{12em} \left. \left. + 2 \hbar^{- 1} \left(
  \int_{\mathbb{R}_{\geqslant 0}} I^{\partial} \star_{\hbar} I^{\partial}
  \right)_{\mathcal{E}_L} \right) \right)\\
  & = & \lim_{\varepsilon \rightarrow 0} e^{\hbar \partial_{P (\varepsilon,
  t)}} \left( e^{(I_0 |_{\mathcal{E}_L}) / \hbar} \left( \mathbb{I} (e^{-
  J^{\partial} / \hbar} \tmmathbf{\Omega}_L^{\tmop{right}} (e^{J^{\partial} /
  \hbar}, I^{\partial})) +  \hbar^{- 1} \left( \int_{\mathbb{R}_{\geqslant
  0}} I^{\partial} \star_{\hbar} I^{\partial} \right)_{\mathcal{E}_L} \right)
  \right) .
\end{eqnarray*}
So we have:

\begin{proposition}
  \label{thm51bvdscrptntqmrpstvlstplsend5151}The QME
  (\ref{bvqmerstctdfldspctqmrplus51}) is equivalent to
  \begin{equation}
    \tmmathbf{\Omega}_L^{\tmop{right}} (e^{J^{\partial} / \hbar},
    I^{\partial}) = 0, \quad \text{and} \quad I^{\partial} \star_{\hbar}
    I^{\partial} = 0. \label{qmeeqvltcdttqmrplsbvfmlsm}
  \end{equation}
\end{proposition}

(\ref{qmeeqvltcdttqmrplsbvfmlsm}) ensures that $(\mathcal{O} (\mathcal{E}_L)
[[\hbar]], \mathd + \hbar \partial_{K_t} + \{ I_t |_{\mathcal{E}_L}, - \}_t)$
is a cochain complex ($\{ -, - \}_t$ is the BV bracket of $\partial_{K_t}$ on
$\mathcal{O} (\mathcal{E}_L)$).

\begin{remark}
  \label{rmk51plsplsendpls}(\ref{qmeeqvltcdttqmrplsbvfmlsm}) also ensures that
  \begin{equation}
    \left( \mathcal{O} (L) [[\hbar]], \frac{1}{\hbar}
    \tmmathbf{\Omega}_L^{\tmop{right}} (-, e^{J^{\partial} / \hbar}
    \star_{\hbar} I^{\partial} \star_{\hbar} e^{- J^{\partial} / \hbar})
    \right) \label{eftwithinbvbint53}
  \end{equation}
  is a derived BV algebra in the sense of {\cite[Definition
  2.1]{bandiera2020cumulants}} (we also refer to {\cite[Section
  5]{wangyantqm202203}} for a review of this notion). If we use homological
  perturbation theory to transfer the operator $\hbar \partial_{K_t} + \{ I_t
  |_{\mathcal{E}_L}, - \}_t$ on $\mathcal{O} (\mathcal{E}_L) [[\hbar]]$ to
  $\mathd$-cohomology of $\mathcal{O} (\mathcal{E}_L) [[\hbar]]$ following
  {\cite[Section 4.1]{wangyantqm202203}}, (\ref{eftwithinbvbint53}) will be
  the resulting ``effective observable complex'' for the current model.
\end{remark}

Similarly, if we consider TQM on $\mathbf{I}$ with restricted bulk field space
$\mathcal{E}_{L_0, L_1}$ defined in (\ref{whatihvnlblthisdfntn44}), the QME
would be
\begin{equation}
  \left( \mathd + \hbar \partial_{K_{L_0, L_1, t}} \right) e^{\left( I_t
  |_{\mathcal{E}_{L_0, L_1}} \right) / \hbar} = 0
  \label{qmetqmrstctdfldspcitvl53}
\end{equation}
where $I_t$ here is defined in (\ref{scltinactntqmintvl49}). Repeating the
arguments for the previous case, we have:

\begin{proposition}
  \label{qmesltntqmintvl56endendend}The QME (\ref{qmetqmrstctdfldspcitvl53})
  is equivalent to
  \begin{equation}
    \tmmathbf{\Omega}_{L_0}^{\tmop{right}} (e^{J^{\partial}_0 / \hbar},
    I^{\partial}) = 0, \quad \tmmathbf{\Omega}_{L_1}^{\tmop{left}}
    (I^{\partial}, e^{J^{\partial}_1 / \hbar}) = 0, \quad \text{and} \quad
    I^{\partial} \star_{\hbar} I^{\partial} = 0
    \label{qmesltnbvrstctdfldspctqmitnvl54}
  \end{equation}
  where $\tmmathbf{\Omega}_{L_1}^{\tmop{left}} (-, -)$ is defined in
  (\ref{weylqtztnccrtfmlmycnvntn}) with the substitution $L \rightarrow L_1',
  L' \rightarrow L_1$.
\end{proposition}

\begin{example}
  \label{eg52savemestpstp}We continue discussing BF theory with B-A boundary
  condition mentioned in Example \ref{eg37thycnstctn}.
  
  With polarization (\ref{bfthyintvlabplrztn}) and $(I^{\partial},
  J^{\partial}_0, J^{\partial}_1)$ in (\ref{bfthydtactnbdractnbfvoprt414}), we
  have $I^{\partial} \star_{\hbar} I^{\partial} = 0$, and
  \begin{eqnarray*}
    \tmmathbf{\Omega}_{L_0}^{\tmop{right}} (e^{J^{\partial}_0 / \hbar},
    I^{\partial}) & = & p_{L_0} \left( e^{\hbar \partial_{(K^{\partial}_{0, +}
    - K^{\partial}_{0, -}) / 4}} I^{\partial} \right) = 0,\\
    \tmmathbf{\Omega}_{L_1}^{\tmop{left}} (I^{\partial}, e^{J^{\partial}_1 /
    \hbar}) & = & p_{L_1} \left( e^{\hbar \partial_{(K^{\partial}_{1, -} -
    K^{\partial}_{1, +}) / 4}} I^{\partial} \right) = - \frac{1}{2} f^{c b}_c
    A_b .
  \end{eqnarray*}
  ($K_{0, -}^{\partial} = K_{1, +}^{\partial} = - \epsilon t_a \otimes \eta
  t^a, K_{0, +}^{\partial} = K_{1, -}^{\partial} = - \eta t^a \otimes \epsilon
  t_a$ here.)
  
  So, if we work on the restricted bulk field space and study BV quantization,
  the B boundary condition will be anomaly free,\footnote{This validates an
  expectation in {\cite[Remark 5.0.3]{rabinovich2021factorization}}.} but
  there will be an anomaly associated to the A boundary condition. If the Lie
  algebra $\mathfrak{g}$ is unimodular (i.e., $f^{c b}_c = 0$), then A
  boundary condition is also anomaly free, and we obtain a consistent theory
  within BV formalism.
\end{example}

\subsubsection*{The mQME revisited}

Now, suppose $(I^{\partial}, J^{\partial})$ satisfies the condition
(\ref{qmeeqvltcdttqmrplsbvfmlsm}), hence defining a TQM on
$\mathbb{R}_{\geqslant 0}$ within BV formalism. The QME solution $I_t
|_{\mathcal{E}_L}$ induces an interactive observable complex
\begin{equation}
  (\mathcal{O} (\mathcal{E}_L) [[\hbar]], \mathd + \hbar \partial_{K_t} + \{
  I_t |_{\mathcal{E}_L}, - \}_t) . \label{glblobscplxtqmrpls55}
\end{equation}
Recall the mQME (\ref{mqmemqmemqme322}) (with $H^{\partial} = -
e^{J^{\partial} / \hbar} \star_{\hbar} I^{\partial} \star_{\hbar} e^{-
J^{\partial} / \hbar}$):
\[ (1 \otimes (\hbar \mathd + \hbar^2 \partial_{K_t})
   +\tmmathbf{\Omega}_{L'}^{\tmop{left}} (H^{\partial}, -) \otimes 1)
   \mathbb{I}_{\theta_t} (e^{(\pi^{\ast} (\tmmathbf{\alpha}) + I_t) / \hbar})
   = 0, \]
it can be rewritten as
\begin{equation}
  (1 \otimes (\hbar \mathd + \hbar^2 \partial_{K_t} + \hbar \{ I_t
  |_{\mathcal{E}_L}, - \}_t) +\tmmathbf{\Omega}_{L'}^{\tmop{left}}
  (H^{\partial}, -) \otimes 1) e^{- (I_t |_{\mathcal{E}_L}) / \hbar}
  \mathbb{I}_{\theta_t} (e^{(\pi^{\ast} (\tmmathbf{\alpha}) + I_t) / \hbar}) =
  0. \label{rwrtngofmqme155}
\end{equation}
It is direct to verify that
\begin{equation}
  e^{- (I_t |_{\mathcal{E}_L}) / \hbar} \mathbb{I}_{\theta_t} (e^{(\pi^{\ast}
  (\tmmathbf{\alpha}) + I_t) / \hbar}) = e^{- (I_t |_{\mathcal{E}_L}) / \hbar}
  \lim_{\varepsilon \rightarrow 0} (1 \otimes e^{\hbar \partial_{P
  (\varepsilon, t)}}) e^{(\mathbb{I} (\tmmathbf{\alpha}) + I_0
  |_{\mathcal{E}_L}) / \hbar} \label{mdleqtnind57lst}
\end{equation}
is of the form $e^{A_t / \hbar}$ with $A_t \in \mathcal{O} (L') \otimes
\mathcal{O} (\mathcal{E}_L) [[\hbar]]$.\footnote{$A_t$ is the summation over
connected Feynman graphs containing vertices corresponding to $\mathbb{I}
(\tmmathbf{\alpha})$.}

Then, we can use the pairing
\[ \ll -, - \gg : \mathcal{O} (L) [[\hbar]] \otimes \mathcal{O} (L') [[\hbar]]
   \rightarrow \mathbb{R} [[\hbar]] \]
defined in (\ref{hilbertpairingnondegnmycnvntn}) to dualize the $\mathcal{O}
(L')$ component of $e^{- (I_t |_{\mathcal{E}_L}) / \hbar}
\mathbb{I}_{\theta_t} (e^{(\pi^{\ast} (\tmmathbf{\alpha}) + I_t) / \hbar})$. A
consequence of (\ref{usefulfactabta2mycvtn1}) is that, the map $\ll -,
e^{\mathbb{I} (\tmmathbf{\alpha}) / \hbar} \gg$ actually equals to
$\mathbb{I}$ in (\ref{nttnspcliso316}), mapping $\mathcal{O} (L) [[\hbar]]$ to
$\mathcal{O} (\mathcal{E}_L) [[\hbar]]$. So by (\ref{mdleqtnind57lst}),
\begin{equation}
  \mathbb{I}_{(0, t)} \assign \ll -, e^{- (I_t |_{\mathcal{E}_L}) / \hbar}
  \mathbb{I}_{\theta_t} (e^{(\pi^{\ast} (\tmmathbf{\alpha}) + I_t) / \hbar})
  \gg = e^{- (I_t |_{\mathcal{E}_L}) / \hbar} \lim_{\varepsilon \rightarrow 0}
  e^{\hbar \partial_{P (\varepsilon, t)}} \left( e^{(I_0 |_{\mathcal{E}_L}) /
  \hbar} \mathbb{I} (-) \right), \label{wvfctnisamp56}
\end{equation}
which also maps $\mathcal{O} (L) [[\hbar]]$ to $\mathcal{O} (\mathcal{E}_L)
[[\hbar]]$. $\mathbb{I}_{(0, t)}$ identifies the observable at scale $t$
corresponding to a given ``boundary observable at scale $0$'' under
renormalization group flow. Then, by (\ref{hilbtprcptbwithmdleactn}), the mQME
(\ref{rwrtngofmqme155}) is equivalent to the condition that
\begin{equation}
  \mathbb{I}_{(0, t)} : \left( \mathcal{O} (L) [[\hbar]], \frac{- 1}{\hbar}
  \tmmathbf{\Omega}_L^{\tmop{right}} (-, H^{\partial}) \right) \rightarrow
  (\mathcal{O} (\mathcal{E}_L) [[\hbar]], \mathd + \hbar \partial_{K_t} + \{
  I_t |_{\mathcal{E}_L}, - \}_t) \label{cchnmpeftbdrtoglblobsvbls57}
\end{equation}
is a cochain map.

There is a natural way to understand this condition. For the TQM induced by
$(I^{\partial}, J^{\partial})$, we can construct a factorization algebra
$\mathcal{F}$ of observables in the sense of
{\cite{costello_gwilliam_2016,costello_gwilliam_2021}}, as described in
{\cite{rabinovich2021factorization}}. This factorization algebra assigns a
cochain complex\footnote{Actually, to each open subset, $\mathcal{F}$ assigns
a family of cochain complex labelled by scales.} $\mathcal{F} (U)$ to each
open subset $U$ of $\mathbb{R}_{\geqslant 0}$, and assigns a cochain map
\[ \mathcal{F} (U_1) \otimes \mathcal{F} (U_2) \otimes \cdots \otimes
   \mathcal{F} (U_n) \rightarrow \mathcal{F} (U) \]
to non-intersecting $U_1, \ldots, U_n$ which all lies in $U$. Particularly,
the ``global observables'' $\mathcal{F} (\mathbb{R}_{\geqslant 0})$ (at scale
$t$) is just (\ref{glblobscplxtqmrpls55}). $\mathcal{F} (U)$ is the subset of
$\mathcal{F} (\mathbb{R}_{\geqslant 0})$ consisting of functionals supported
on $U$.\footnote{Honestly speaking, elements in $\mathcal{F} (U)$ at scale $t$
may have supports larger than $U$. We refer to {\cite[Definition
4.7.5]{rabinovich2021factorization}} for precise definition.}

For nested open subsets $[0, x_1) \subset [0, x_2) \subset \cdots \subset
\mathbb{R}_{\geqslant 0}$, the factorization algebra data leads to the
following sequence of embeddings:
\[ \mathcal{F} ([0, x_1)) \hookrightarrow \mathcal{F} ([0, x_2))
   \hookrightarrow \cdots \hookrightarrow \mathcal{F} (\mathbb{R}_{\geqslant
   0}) = (\mathcal{O} (\mathcal{E}_L) [[\hbar]], \mathd + \hbar \partial_{K_t}
   + \{ I_t |_{\mathcal{E}_L}, - \}_t) . \]
We can imagine a limit among all such $\mathcal{F} ([0, x))$'s, which should
consist of functionals supported on $\{ 0 \}$. {\cite[Section
10.1]{costello_gwilliam_2021}} contains a rigorous formulation on this kind of
limit based on the factorization algebra data, which we do not explain here
for brevity. By the content of (\ref{wvfctnisamp56}), it is convincing that
\[ \left( \mathcal{O} (L) [[\hbar]], \frac{- 1}{\hbar}
   \tmmathbf{\Omega}_L^{\tmop{right}} (-, H^{\partial}) \right) \]
will be the expected limit of these $\mathcal{F} ([0, x))$'s if we make things
precise. Besides, $\mathbb{I}_{(0, t)}$ should be the limit of the ``local to
global maps'' $\mathcal{F} ([0, x)) \hookrightarrow \mathcal{F}
(\mathbb{R}_{\geqslant 0})$, hence (\ref{cchnmpeftbdrtoglblobsvbls57}) follows
from axioms of factorization algebra.

As for TQM on interval, the story is similar. Suppose $(I^{\partial},
J^{\partial}_0, J^{\partial}_1)$ satisfies
(\ref{qmesltnbvrstctdfldspctqmitnvl54}), we have an interactive observable
complex
\[ \left( \mathcal{O} (\mathcal{E}_{L_0, L_1}) [[\hbar]], \mathd + \hbar
   \partial_{K_{L_0, L_1, t}} + \left\{ I_t |_{\mathcal{E}_{L_0, L_1}}, -
   \right\}_t \right) \]
induced by the solution $I_t |_{\mathcal{E}_{L_0, L_1}}$ to QME
(\ref{qmetqmrstctdfldspcitvl53}). Similar to (\ref{wvfctnisamp56}), we define
a map
\begin{equation}
  \mathbb{I}_{(0, t)} \assign e^{- \left( I_t |_{\mathcal{E}_{L_0, L_1}}
  \right) / \hbar} \lim_{\varepsilon \rightarrow 0} e^{\hbar \partial_{P_{L_0,
  L_1} (\varepsilon, t)}} \left( e^{\left( I_0 |_{\mathcal{E}_{L_0, L_1}}
  \right) / \hbar} m_{\mathcal{O} (\mathcal{E}_{L_0, L_1})} (\mathbb{I}_0
  \otimes \mathbb{I}_1 (-)) \right) \label{tqmintvlbdrobsvbltoscltblkobs59}
\end{equation}
from $\mathcal{O} (L_0) \otimes \mathcal{O} (L_1) [[\hbar]]$ to $\mathcal{O}
(\mathcal{E}_{L_0, L_1}) [[\hbar]]$, where $I_0 |_{\mathcal{E}_{L_0, L_1}}
=\mathbb{I}_0 (J^{\partial}_0) +\mathbb{I}_1 (J^{\partial}_1) + \left(
\int_{\mathbf{I}} I^{\partial} \right)_{\mathcal{E}_{L_0, L_1}}$,
$\mathbb{I}_0, \mathbb{I}_1$ are defined in (\ref{lstfmlisweargdfk491003}) and
$m_{\mathcal{O} (\mathcal{E}_{L_0, L_1})} : \mathcal{O} (\mathcal{E}_{L_0,
L_1})^{\otimes 2} \rightarrow \mathcal{O} (\mathcal{E}_{L_0, L_1})$ denotes
the symmetric product on $\mathcal{O} (\mathcal{E}_{L_0, L_1})$. Then, the
mQME (\ref{mqmeintvltqmmycvtn49})
\begin{eqnarray*}
  &  & ((\tmmathbf{\Omega}_{L_0'}^{\tmop{left}} (H^{\partial}_0, -) \otimes 1
  + 1 \otimes \tmmathbf{\Omega}_{L_1'}^{\tmop{right}} (-, H_1^{\partial}))
  \otimes 1 \nobracket\\
  &  & \hspace{8em} \left. + 1 \otimes 1 \otimes \left( \hbar \mathd +
  \hbar^2 \partial_{K_{L_0, L_1, t}} \right) \right) \mathbb{I}_{\theta_t}
  (e^{(\pi^{\ast}_0 (\tmmathbf{\alpha}_0) - \pi^{\ast}_1 (\tmmathbf{\alpha}_1)
  + I_t) / \hbar}) = 0
\end{eqnarray*}
(with $H_0^{\partial} = - e^{J^{\partial}_0 / \hbar} \star_{\hbar}
I^{\partial} \star_{\hbar} e^{- J_0^{\partial} / \hbar}, H_1^{\partial} = e^{-
J^{\partial}_1 / \hbar} \star_{\hbar} I^{\partial} \star_{\hbar}
e^{J_1^{\partial} / \hbar}$) is equivalent to the condition that
(\ref{tqmintvlbdrobsvbltoscltblkobs59}) is a cochain map:
\begin{eqnarray}
  \mathbb{I}_{(0, t)} & : & \left( \mathcal{O} (L_0) [[\hbar]], \frac{-
  1}{\hbar} \tmmathbf{\Omega}_{L_0}^{\tmop{right}} (-, H_0^{\partial}) \right)
  \otimes \left( \mathcal{O} (L_1) [[\hbar]], \frac{- 1}{\hbar}
  \tmmathbf{\Omega}_{L_1}^{\tmop{left}} (H_1^{\partial}, -) \right)
  \rightarrow \nonumber\\
  &  & \left( \mathcal{O} (\mathcal{E}_{L_0, L_1}) [[\hbar]], \mathd + \hbar
  \partial_{K_{L_0, L_1, t}} + \left\{ I_t |_{\mathcal{E}_{L_0, L_1}}, -
  \right\}_t \right) .  \label{the59mpisacchmp510}
\end{eqnarray}
Just like the previous case, suppose $\mathcal{F}$ is the factorization
algebra constructed for the current TQM on interval, then the map
$\mathbb{I}_{(0, t)}$ should be regarded as the limit among all such ``local
to global maps'':
\[ \mathcal{F} ([0, x)) \otimes \mathcal{F} ((x', 1]) \rightarrow \mathcal{F}
   (\mathbf{I}) = \left( \mathcal{O} (\mathcal{E}_{L_0, L_1}) [[\hbar]],
   \mathd + \hbar \partial_{K_{L_0, L_1, t}} + \left\{ I_t
   |_{\mathcal{E}_{L_0, L_1}}, - \right\}_t \right) \]
associated to $[0, x) \sqcup (x', 1] \hookrightarrow [0, 1]$. This explains
mQME in the framework of
{\cite{costello_gwilliam_2016,costello_gwilliam_2021}}.

\begin{remark}
  As mentioned in Remark \ref{rmk51plsplsendpls}, we can use homological
  perturbation theory to construct a projection from $\mathcal{O}
  (\mathcal{E}_L) [[\hbar]]$ (or $\mathcal{O} (\mathcal{E}_{L_0, L_1})
  [[\hbar]]$) to the effective observable complex. Then, by composing
  $\mathbb{I}_{(0, t)}$ with this projection, we obtain a map from
  $\mathcal{O} (L) [[\hbar]]$ (or $\mathcal{O} (L_0) \otimes \mathcal{O} (L_1)
  [[\hbar]]$) to the effective observable complex. This map is the ``state''
  from the factorization algebra perspective. (Moreover, if we go through the
  calculations in Remark \ref{rmk51plsplsendpls} for TQM on
  $\mathbb{R}_{\geqslant 0}$, this ``state'' will be the identity map on
  $\mathcal{O} (L) [[\hbar]]$, reflecting the trivial time evolution of TQM.)
  
  Accordingly, we can take $\mathbb{I}_{\theta_t} (e^{(\pi^{\ast}
  (\tmmathbf{\alpha}) + I_t) / \hbar})$ in (\ref{mqmemqmemqme322}) (or
  $\mathbb{I}_{\theta_t} (e^{(\pi^{\ast}_0 (\tmmathbf{\alpha}_0) -
  \pi^{\ast}_1 (\tmmathbf{\alpha}_1) + I_t) / \hbar})$ in
  (\ref{mqmeintvltqmmycvtn49})) and project its $\mathcal{O} (\mathcal{E}_L)$
  (or $\mathcal{O} (\mathcal{E}_{L_0, L_1})$) component to the
  $\mathd$-cohomology using homological perturbation theory. The outcome is
  the ``state'' in perturbative BV-BFV formalism, which mimics physicists'
  ``wave function'' description.
  
  So, our translation of mQME connects these two interpretations of the
  ``state''.
\end{remark}

Yau Mathematical Sciences Center, Tsinghua University, Bejing 100084, P.R.
China

Email: wmh18@mails.tsinghua.edu.cn

Institute for Advanced Study, Tsinghua University, Bejing 100084, P.R.
China

Email: ygw17@mails.tsinghua.edu.cn

\end{document}